\documentclass[reqno]{amsart}
\usepackage{amsaddr}

\usepackage[english]{babel}
\usepackage[latin1]{inputenc} 
\usepackage[T1]{fontenc}

\usepackage{amsfonts,amssymb,amsmath,amsthm}
\usepackage{stmaryrd}
\numberwithin{equation}{section}
\usepackage{dsfont}
\usepackage{multirow}
\usepackage{url}
\usepackage{algorithmic,algorithm}
\usepackage{graphics}
\usepackage{booktabs}
\usepackage{epsfig}
\usepackage{epstopdf}
\usepackage{psfrag}               
\usepackage{mathrsfs}
\usepackage{mathtools}            
\usepackage{subfigure}
\usepackage{longtable}	
\usepackage[round]{natbib}
\usepackage{paralist}
\usepackage{comment}
\usepackage[inline]{enumitem}
\usepackage{color}

\usepackage{hyperref}

\newcommand{\bbE}{\mathbb{E}}

\newcommand{\bbN}{\mathbb{N}}

\newcommand{\bbR}{\mathbb{R}}

\newcommand{\cC}{\mathcal{C}}
\newcommand{\cD}{\mathcal{D}}

\newcommand{\cH}{\mathcal{H}}

\newcommand{\cO}{\mathcal{O}}
\newcommand{\cP}{\mathcal{P}}
\newcommand{\cQ}{\mathcal{Q}}

\newcommand{\cT}{\mathcal{T}}

\newcommand{\cW}{\mathcal{W}}

\newcommand{\scrD}{\mathscr{D}}

\newcommand{\kvec}{\mathbf{k}}
\newcommand{\svec}{\mathbf{s}}

\newcommand{\xivec}{\boldsymbol{\xi}}

\newcommand{\norm}[2]{     \| #1       \|_{ #2 }}

\newcommand{\normiii}[2]{\vert\kern-0.25ex\vert\kern-0.25ex\vert #1 \vert\kern-0.25ex \vert\kern-0.25ex\vert_{ #2 }}
\newcommand{\Normiii}[2]{\left\vert\kern-0.25ex\left\vert\kern-0.25ex\left\vert #1 \right\vert\kern-0.25ex\right\vert\kern-0.25ex\right\vert_{ #2 }}
\newcommand{\scalar}[2]{     ( #1       )_{ #2 }}

\newcommand{\white}{\cW}

\newcommand{\rd}{\mathrm{d}}
\newcommand{\from}{\colon}

\DeclarePairedDelimiter\ceil{\lceil}{\rceil}
\DeclarePairedDelimiter\floor{\lfloor}{\rfloor}

\newcommand{\mv}[1]{{\boldsymbol{\mathrm{#1}}}}
\newcommand{\trsp}{\ensuremath{\top}}

\newcommand{\proper}{\mathsf}

\newcommand{\pN}{\proper{N}}

\newcommand{\lop}{P_{\ell,h}}
\newcommand{\rop}{P_{r,h}}
\newcommand{\lopex}{P_\ell}
\newcommand{\ropex}{P_r}
\newcommand{\lmat}{\mathbf{P}_\ell}
\newcommand{\rmat}{\mathbf{P}_r}

\graphicspath{{Images/}}

\newtheorem{lemma}{Lemma}[section]
\newtheorem{proposition}[lemma]{Proposition}
\newtheorem{theorem}[lemma]{Theorem}

\theoremstyle{remark}
\newtheorem{remark}[lemma]{Remark}

\theoremstyle{definition}

\definecolor{dbcolor}{RGB}{0,0,0}
\definecolor{kkcolor}{RGB}{0,0,0}
\newcommand{\db}[1]{{\color{dbcolor}#1}}
\newcommand{\kk}[1]{{\color{kkcolor}#1}}
\usepackage[normalem]{ulem}

\usepackage{bbm}

\begin{document}

\title[The rational SPDE approach for Gaussian random fields]
	{The rational SPDE approach for Gaussian random fields with general smoothness}


\author{David Bolin$^a$ and Kristin Kirchner$^{a,b}$}
\address[D.~Bolin and K.~Kirchner]{
$^A$Department of Mathematical Sciences, Chalmers University of Technology and University of Gothenburg, G\"oteborg, Sweden; \\
$^B$Seminar for Applied Mathematics, ETH Z\"urich, Switzerland}

\email{{\normalfont{(D.~Bolin, corresponding author) }}david.bolin@kaust.edu.sa}


\thanks{
This work has been supported in part by the Swedish Research Council under grant
No.\ 2016-04187 and the Knut and Alice Wallenberg Foundation (KAW 20012.0067).
The authors thank Stig Larsson and Finn Lindgren for valuable comments on the manuscript.
}


\begin{abstract}
A popular approach for modeling and inference in spatial statistics is to represent Gaussian random fields as solutions to stochastic partial differential equations (SPDEs) of the form $L^{\beta}u = \white$, 
where $\white$ is Gaussian white noise, 
$L$ is a second-order differential operator, 
and $\beta>0$ is a parameter that determines the smoothness of $u$.
However, this approach has been limited 
to the case $2\beta\in\mathbb{N}$,
which excludes several important models 
and makes it necessary to keep $\beta$ fixed during inference.

We propose a new method, the rational SPDE approach, which in spatial dimension $d\in\bbN$ is applicable for any \kk{$\beta>d/4$}, and
thus remedies the mentioned limitation.
The presented scheme combines a finite element discretization with a rational approximation of the function $x^{-\beta}$ 
to approximate $u$.
For the resulting approximation,
an explicit rate of convergence 
to $u$ \kk{in mean-square sense is derived. 
Furthermore,} we show that \kk{our} method has
the same computational benefits
as in the restricted case $2\beta\in\bbN$.
Several numerical experiments and a statistical application 
are used to illustrate the accuracy of the method,
and to show that it facilitates likelihood-based inference for all model parameters including $\beta$. 
\end{abstract}




\maketitle

\section{Introduction}

One of the main challenges in spatial statistics is to handle large data sets.
A reason for this is that the computational cost for likelihood evaluation and spatial prediction
is in general cubic in the number $N$ of observations of a Gaussian random field.
A tremendous amount of research has been devoted
to coping with this problem
and various methods have been suggested
\citep[see][for a recent review]{heaton2017methods}.

A common approach is to define an approximation $u_h$ of a Gaussian random field $u$
on a spatial domain $\cD$ via a basis expansion,
\begin{equation}\label{e:basisexp}
    u_h(\mv{s}) = \sum_{j=1}^n u_j \, \varphi_j(\mv{s}),
    \qquad
    \mv{s}\in\cD,
\end{equation}
where 
$\varphi_j \from \cD \to \bbR$ are fixed basis functions
and $\mv{u} = (u_1,\ldots, u_n)^{\trsp} \sim \pN(\mv{0},\mv{\Sigma}_{\mv{u}})$
are stochastic weights.
The computational effort can then be reduced by choosing \kk{$n\ll N$. 
However,} methods based on such low-rank 
approximations tend to remove fine-scale variations of the process.
%
For this reason, methods which instead exploit sparsity
for reducing the computational cost have gained popularity in recent years.
One can construct sparse approximations either of the covariance matrix of the measurements \citep{furrer2006covariance},
or of the inverse of the covariance matrix \citep{datta2016hierarchical}.
Alternatively, one can let the precision matrix~$\mv{\Sigma}_{\mv{u}}^{-1}$
of the weights in~\eqref{e:basisexp} be sparse,
as in the stochastic partial differential equation (SPDE) approach
by \cite{lindgren11}, where usually $n\approx N$. 
To increase the accuracy \db{further}, several combinations of
the methods mentioned above have been considered \citep[e.g.,][]{sang2012full}
and multiresolution approximations of the process
have been exploited
\citep{nychka2015multiresolution, katzfuss2017multi}.
However, \db{theoretical error bounds have not been derived for most of these methods},
which necessitates tuning these approximations for each specific model.

In this work we propose a new class of approximations,
\kk{whose members we refer to as \emph{rational 
stochastic partial differential equation approximations} 
or \emph{rational approximations} for short}.
Our approach is similar to some of the above methods
in the sense that an expansion \eqref{e:basisexp}
with compactly supported basis functions is exploited.
The main novelty is that neither the covariance matrix $\mv{\Sigma}_{\mv{u}}$
nor the precision matrix $\mv{\Sigma}_{\mv{u}}^{-1}$ of the weights $\mv{u}$
\kk{is} assumed to be sparse.
The covariance matrix is instead a product
\kk{$\mv{\Sigma}_{\mv{u}} = \mv{P} \mv{Q}^{-1}\mv{P}^{\trsp}$,
where $\mv{P}$ and $\mv{Q}$ are sparse matrices
and the sparsity pattern of $\mv{P}$} is a subset
of \kk{that} of~$\mv{Q}$.
We show that the resulting approximation
facilitates inference and prediction
at the same computational cost as a Markov approximation
with $\mv{\Sigma}_{\mv{u}}^{-1} = \mv{Q}$, and at a higher accuracy.

For the theoretical framework of our approach,
we consider a Gaussian random field on a bounded domain $\cD \subset \bbR^d$
which can be represented as the solution $u$ to the SPDE
\begin{align}\label{e:Lbeta}
    L^\beta u =  \white \quad\text{in }\cD,
\end{align}
where $\white$ is Gaussian white noise on $\cD$,
and $L^{\beta}$ is a fractional power of
a second-order differential operator $L$
which determines the covariance \kk{structure} of $u$.
Our \kk{rational approximations are} 
based on two components:
\begin{enumerate*}[label=(\roman*)]
\item a finite element method (FEM)
    with continuous and piecewise polynomial basis functions $\{\varphi_j\}_{j=1}^n$, and
\item a rational approximation of the function $x^{-\beta}$.
\end{enumerate*}
We explain how to perform these two steps in practice
in order to explicitly compute the matrices $\mv{P}$ and $\mv{Q}$.
Furthermore, we derive an upper bound
for the strong mean-square error of the rational approximation.
This bound provides an explicit rate of convergence in terms
of the mesh size of the finite element discretization,
which facilitates tuning the approximation
without empirical tests
for each specific model.

Examples of random fields
which can be expressed as solutions
to SPDEs of the form \eqref{e:Lbeta}
include approximations of Gaussian Mat\'ern fields \citep{matern60}.
Specifically, if $\cD:=\bbR^d$ 
a zero-mean Gaussian Mat\'ern field can be viewed
as a solution $u$ to
\begin{equation}\label{e:statmodel}
        (\kappa^2 - \Delta)^{\beta} \, (\tau u) = \white,
\end{equation}
where $\Delta$ denotes the Laplacian \citep{whittle63}.
The constant parameters $\kappa, \tau > 0$
determine the practical correlation range
and the variance of $u$.
The exponent~$\beta$ defines the smoothness parameter
$\nu$ of the Mat\'ern covariance function
via the relation $2\beta= \nu + d/2$
and, thus, the differentiability of the field. 
\kk{For applications, variance, range and differentiability 
typically are the most important properties 
of a Gaussian field.
For this reason, the Mat\'ern model is 
highly popular in spatial statistics 
and has become the} \db{standard choice 
for Gaussian process priors 
in machine learning \citep{Rasmussen2006}.} 
\kk{Since \eqref{e:statmodel} 
is a special case of \eqref{e:Lbeta} 
we believe that 
the outcomes of this work 
will be of great relevance
for many statistical applications, 
see also \S\ref{sec:application}}.

In contrast to covariance-based models,
\db{the SPDE} \kk{approach additionally} has the advantage
that it allows for a number of generalizations
of stationary Mat\'ern fields including
\begin{enumerate*}[label=(\roman*)]
    \item non-stationary fields 
        \kk{generated by non-stationary}  
        differential operators 
        \citep[e.g.,][]{fuglstad2015non-stationary},
    \item fields on more general domains 
    	such as the sphere 
    	\citep[e.g.,][]{lindgren11}, and
    \item non-Gaussian \kk{models} 
    	\citep[e.g.,][]{wallin15}.
\end{enumerate*}

\cite{lindgren11} showed
that, if $2\beta\in\bbN$, one can construct accurate approximations of the form \eqref{e:basisexp}
for Gaussian Mat\'ern fields 
on bounded domains $\cD\subsetneq\bbR^d$,
such that $\mv{\Sigma}_{\mv{u}}^{-1}$ is sparse.
To this end, \eqref{e:statmodel} is considered on $\cD$ 
and the differential operator $\kappa^2-\Delta$ is augmented
with appropriate boundary conditions.
The resulting \kk{SPDE} is then solved
approximately by means of a FEM. 
\kk{Due to} 
\db{the implementation in the R-INLA software, 
this \kk{approach} has become \kk{widely used}, 
see \citep{bakka2018} for a \kk{comprehensive} 
list of recent applications}. 

%
However, the restriction $2\beta\in\bbN$ implies 
a significant limitation for the
flexibility of the method.
In particular, it is therefore not
directly applicable to the important special case
of exponential covariance ($\nu=1/2$) on $\bbR^2$,
where $\beta=3/4$.
In addition, restricting the value of $\beta$
complicates \kk{estimating} 
the smoothness of the process from data. \kk{In fact, $\beta$ typically is fixed
when the method is used in practice, 
since identifying the value of $2\beta \in \bbN$ 
with the highest likelihood 
requires a separate estimation of
all the other parameters in the model 
for each value of $\beta$}. 
A common justification for fixing $\beta$ is to argue
that it is not practicable to estimate 
the smoothness of a random field from data.
\db{However, there are certainly applications for which 
it is feasible to estimate the smoothness. We provide an example of this in~\S\ref{sec:application}}.
\db{Furthermore, having the correct smoothness of the model is particularly important for interpolation, and the fact that the Mat\'ern model allows for estimating the smoothness from data was the main reason for why
\cite{stein99} recommended the model.}

%

The rational SPDE approach 
\kk{presented in this work}
facilitates an estimation of $\beta$
from data by providing an approximation of $u$
which is computable for all fractional powers \kk{$\beta > d/4$ (i.e., $\nu>0$), 
where $d\in\bbN$ is the dimension of the spatial 
domain $\cD\subset\bbR^d$}.
It thus enables to include this parameter in
likelihood-based (or Bayesian) parameter estimation 
for both stationary and non-stationary models.
Although the SPDE approach has been considered
in the non-fractional case also for non-stationary models,
\cite{lindgren11} showed convergence of the approximation
only for the stationary model \eqref{e:statmodel}. 
Our analysis \kk{in \S\ref{sec:rational}} closes this gap 
since we consider the general model \eqref{e:Lbeta}
which covers the non-stationary case 
\db{and several other previously proposed 
\kk{generalizations} of the Mat\'ern model}. 

The structure of this article is  as follows:
We briefly review existing methods
for the SPDE approach in the fractional case in \S\ref{sec:SPDE}.
In \S\ref{sec:rational} the rational SPDE approximation
is introduced and a result on its strong convergence is stated.
The procedure of \kk{applying 
the rational SPDE approach to} statistical inference
is addressed in \S\ref{sec:comp}.
\S\ref{sec:numerics} contains numerical experiments
which illustrate the accuracy of the proposed method.
\db{The identifiability of the parameters in the Mat\'ern SPDE model \eqref{e:statmodel} is discussed in \S\ref{sec:inference}, where we derive necessary and sufficient conditions for equivalence of the induced Gaussian measures.}
In \S\ref{sec:application} we present an application
to climate data,
where we consider fractional and non-fractional models
for both stationary and non-stationary covariances.
We conclude with a discussion in \S\ref{sec:discussion}.
Finally, \db{the supplementary} \kk{material contains} 
four appendices providing details
about \begin{enumerate*}[label=(\Alph*)]
\item the finite element discretization,
\item the convergence analysis, 
\item a comparison with the quadrature method by \cite{bolin2017numerical}, and
\item the equivalence of Gaussian measures. 
\end{enumerate*}
The method developed in this work has been implemented
in the R \citep{Rlanguage} package rSPDE \citep{rSPDEpackage}.

\section{The SPDE approach in the fractional case until now}\label{sec:SPDE}

A reason for why the approach by \cite{lindgren11}
only works for integer values of $2\beta$
is given by \cite{rozanov1977markov},
who showed that a Gaussian random field
on $\bbR^d$ is Markov
if and only if its spectral density
can be written as the reciprocal of a polynomial,
$\widetilde{S}(\kvec) = (2\pi)^{-d}(\sum_{j=0}^m b_j \|\kvec\|^{2j})^{-1}$.
Since the spectrum of a Gaussian Mat\'ern field is
\begin{equation}\label{eq:maternSpec}
    S(\kvec) = \frac{1}{(2\pi)^d}\frac{1}{(\kappa^2 + \|\kvec\|^2)^{2\beta}},
    \qquad
    \kvec\in\bbR^d,
\end{equation}
the precision matrix $\mv{Q}$ will therefore
not be sparse unless $2\beta\in\mathbb{N}$. For $2\beta\notin\mathbb{N}$,
\cite{lindgren11} suggested to compute a Markov approximation
by choosing $m = \ceil{2\beta}$ and selecting the coefficients
$\mv{b} = (b_1, \ldots, b_m)^{\trsp}$
so that the deviation between the spectral densities
$\int_{\bbR^d} w(\kvec)(S(\kvec) - \widetilde{S}(\kvec))^2 \,\rd\kvec$
is minimized.
For this measure of deviation, $w$ is some suitable weight function
which should be chosen to get a good approximation
of the covariance function.
For the method to be useful in practice,
the coefficients $b_j$ should be given explicitly in terms
of the parameters~$\kappa$ and $\nu$.
Because of this, \cite{lindgren11}
proposed a weight function that enables an analytical evaluation of the integral,
\begin{align*}
    \int_{\kappa^2}^{\infty}\Bigl[ z^{2\beta} - \sum_{j=0}^m b_j (z-\kappa^2)^{j} \Bigr]^2 z^{-2m -1 - \theta} \, \rd z,
\end{align*}
where $\theta > 0$ is a tuning parameter.
By differentiating this integral with respect to the parameters and
setting the differentials equal to zero,
a system of linear equations is obtained,
which can be solved to find the coefficients $\mv{b}$.
The resulting approximation depends strongly on $\theta$,
and one could use numerical optimization to find a good value of $\theta$
for a specific value of $\beta$, or use the choice $\theta = 2\beta - \floor{2\beta}$,
which approximately minimizes the maximal distance between the covariance functions \citep{lindgren11}.
This method was used for the comparison in \citep{heaton2017methods}, and we will use it
as a baseline method when analyzing the accuracy of the rational SPDE approximations in later sections.

Another Markov approximation based on the spectral density
was proposed by \cite{roininen2014sparse}.
These Markov approximations may be sufficient in certain applications;
however, 
any approach based on the spectral density
or the covariance function is difficult to generalize
to models on more general domains than~$\mathbb{R}^d$,
non-stationary models, or non-Gaussian models.
Thus, such methods cannot be used if the full potential
of the SPDE approach should be kept for fractional values of $\beta$.


There is a rich literature on methods for solving deterministic fractional PDEs
\citep[e.g.,][]{bonito2015, gavrilyuk2004, jin2015, nochetto2015},
and some of the methods that have been proposed could be used
to compute approximations of the solution to the SPDE~\eqref{e:statmodel}.
However, any deterministic problem becomes
more sophisticated when randomness is included.
Even methods developed specifically for sampling solutions to
SPDEs like~\eqref{e:statmodel} may be 
\db{difficult to use directly for statistical applications, \kk{when}  
likelihood evaluations, spatial predictions \kk{or} posterior sampling
are needed}.
\kk{For instance, it has been unclear  
if the sampling approach by \cite{bolin2017numerical}, which is
based on a quadrature approximation
for an integral representation of the
fractional inverse $L^{-\beta}$, 
could be used for statistical inference.}
\db{In Appendix~\ref{subsec:rat-comparequad} we show 
that it can be viewed as a 
(computationally less efficient) version 
of the rational SPDE approximations developed in this work. 
\kk{Consequently}, the results in \S\ref{sec:comp} 
on how to use the rational SPDE approach 
for inference \kk{apply} also to that method. 
In \S\ref{subsec:numerics-matern} 
we compare the performance of 
the two methods in practice within the scope of a numerical experiment.}


\section{Rational approximations for fractional SPDEs}\label{sec:rational}

In this section we propose an explicit scheme for
approximating solutions to a class of SPDEs including~\eqref{e:statmodel}.
Specifically, in 
\kk{\S\ref{subsec:fractional}--\S\ref{subsec:discrete}} 
we introduce
the fractional order equation of interest
as well as its finite element discretization.
In \S\ref{subsec:rat-approx} 
we propose a non-fractional equation, whose solution
after specification of certain coefficients
approximates the random field of interest.
For this approximation, we provide 
\kk{a rigorous error bound}
in \S\ref{subsec:error}. 
Finally, in \S\ref{subsec:rat-coeff} we address the computation of the
coefficients in the rational approximation.

\subsection{The fractional order equation}\label{subsec:fractional}

\kk{With the objective of allowing} 
for more general Gaussian random fields 
than the Mat\'ern \kk{class},
we consider the fractional order equation \eqref{e:Lbeta}, 
where $\cD \subset \bbR^d$, $d\in\{1,2,3\}$,
is \kk{an open}, bounded, convex \kk{polytope}, 
\kk{with closure $\overline{\cD}$}, and 
$\white$ is Gaussian white noise in \kk{$L_2(\cD)$.
Here and below, $L_2(\cD)$ is the Lebesgue space
of square-integrable real-valued functions,
which is equipped with the inner product
$\scalar{w,v}{L_2(\cD)} := \int_{\cD} w(\mv{s}) v(\mv{s}) \, \mathrm{d} \mv{s}$.
The Sobolev space of order $k\in\bbN$ is denoted by
$H^k(\cD) := \left\{ w \in L_2(\cD) : 
D^{\gamma} w \in L_2(\cD) \ \forall \, |\gamma|\leq k \right\}$
and $H^1_0(\cD)$ is the subspace of $H^1(\cD)$ containing
functions with vanishing trace.

We assume that the operator 
$L\from\mathscr{D}(L)\subset L_2(\cD) \to L_2(\cD)$
is a linear second-order 
differential operator in divergence form,
\begin{equation}\label{e:L-div} 
	L u = - \nabla \cdot(\mv{H} \nabla u) + \kappa^2 u,
\end{equation}
whose domain of definition $\mathscr{D}(L)$ 
depends on the choice of boundary conditions 
on $\partial\cD$. 
Specifically, we impose homogeneous Dirichlet 
or Neumann boundary conditions 
and set $V=H^1_0(\cD)$ or $V=H^1(\cD)$, 
respectively. 
Furthermore, we let 
the functions $\mv{H}$ and $\kappa$ 
in~\eqref{e:L-div} 
satisfy the following assumptions: 
\begin{enumerate}[label=\Roman*.] 
	\item\label{ass:coeff-H} 
	$\mv{H}\from\overline{\cD}\to\bbR^{d\times d}$ 
	is symmetric, 
	Lipschitz continuous on the closure $\overline{\cD}$, i.e.,
	\[
	\exists C_{\operatorname{Lip}} > 0 : 
	\quad 
	| H_{ij}(\mv{s}) - H_{ij}(\mv{s'}) |
	\leq 
	C_{\operatorname{Lip}} \norm{\mv{s} - \mv{s'}}{} 
	\qquad 
	\forall \mv{s}, \mv{s'}\in \overline{\cD}, 
	\quad 
	\forall i,j\in\{1,\ldots,d\}, 
	\] 
	and
	uniformly positive definite, i.e.,
	\[ 
	\exists C_0 > 0 :
	\quad
	\operatorname{ess} \inf_{\mv{s}\in\cD} \xivec^\trsp \mv{H}(\mv{s}) \xivec 
	\geq C_0 \|\xivec \|^2
	\qquad
	\forall \xivec \in \bbR^d;  
	\]  
	\item\label{ass:coeff-kappa}  
	$\kappa \from \cD \to \bbR$ 
	is essentially bounded, 
	$\kappa\in L_{\infty}(\cD)$. Furthermore, if Neumann boundary conditions are imposed, then 
	$
	\operatorname{ess} \inf_{\mv{s}\in\cD} \kappa(\mv{s}) \geq \kappa_0 > 0
	$  holds.
\end{enumerate}
If \ref{ass:coeff-H}--\ref{ass:coeff-kappa} 
are satisfied, the differential operator 
$L$ in~\eqref{e:L-div} induces a symmetric, 
continuous and coercive   
bilinear form $a_L$ on $V$, 
\begin{equation}\label{e:a-L} 
	a_L \from V\times V \to \bbR, 
	\qquad 
	a_L (u,v) 
	:= 
	(\mv{H}\nabla u, \nabla v)_{L_2(\cD)}
	+ 
	(\kappa^2 u, v)_{L_2(\cD)}, 
\end{equation} 
and its domain  
is given by 
$\mathscr{D}(L) = 
H^2(\cD)\cap V$}. 
%
\kk{Furthermore, 
Weyl's law 
\citep[see, e.g.,][Thm.~6.3.1]{Davies1995} 
shows that the eigenvalues $\{\lambda_j\}_{j\in\bbN}$ 
of the elliptic differential operator~$L$ 
in~\eqref{e:L-div}, 
in nondecreasing order, 
satisfy the spectral asymptotics 
\begin{align}\label{e:lambdaj}
\lambda_j \eqsim j^{2/d}
\qquad
\text{as }j\to\infty.
\end{align}
Thus}, \db{existence and uniqueness of the solution $u$ 
to \eqref{e:Lbeta} 
\kk{readily} follow from Lemma~2.1 and Proposition~2.3 
of \cite{bolin2017numerical}. We formulate this as a proposition.}
\begin{proposition}\label{prop:regularity}
\db{Let $L$ be given by \eqref{e:L-div} where $\mv{H}$ and $\kappa$ 
\kk{satisfy} the assumptions \ref{ass:coeff-H}--\ref{ass:coeff-kappa} 
above and assume $\beta > d/4$. 
Then \eqref{e:Lbeta} has a \kk{unique solution $u$ 
in $L_2(\Omega;L_2(\cD))$}.}
\end{proposition}

\kk{The assumptions \ref{ass:coeff-H}--\ref{ass:coeff-kappa} 
on the differential operator $L$ 
are satisfied, e.g., by the 
Mat\'ern operator $L = \kappa^2 - \Delta$, 
in which case the condition $\beta>d/4$ 
on the fractional exponent in~\eqref{e:Lbeta}
corresponds to a positive smoothness parameter $\nu$, 
i.e., to a non-degenerate field. 
Moreover, the equation~\eqref{e:Lbeta} 
as considered in our work 
includes several previously proposed non-fractional 
non-stationary models as special cases, 
such as} \db{the non-stationary Mat\'ern models 
by \cite{lindgren11}, 
the models with locally varying anisotropy 
by \cite{fuglstad2015non-stationary}, 
and the barrier models by \cite{bakka2019}. 
Thus, Proposition~\ref{prop:regularity} shows existence 
and uniqueness of the fractional versions of all these models, 
which can be treated \kk{in practice by} 
using the results of the following sections}.

\subsection{\kk{The discrete model}}\label{subsec:discrete}

In order to discretize the problem,
we assume that 
$V_h \subset V$ is a finite element space with continuous
piecewise \kk{linear} basis functions $\{\varphi_{j}\}_{j=1}^{n_h}$
defined with respect to
a triangulation $\cT_h$ of the domain $\overline{\cD}$
indexed by the mesh \kk{width}
$h := \max_{T\in\cT_h} h_T$, where $h_T := \operatorname{diam}(T)$ 
\db{is the diameter of the element $T\in\cT_h$}.
Furthermore, the family $(\cT_h)_{h\in(0,1)}$ of triangulations
inducing the finite-dimensional subspaces $(V_h)_{h\in(0,1)}$ of $V$
is supposed to be quasi-uniform, i.e., 
there exist constants $C_1, C_2 > 0$ such that
$\rho_T \geq C_1 h_T$ and
$h_T \geq C_2 h$ for all
$T\in\cT_h$ and $h \in (0,1)$.
Here,
\kk{$\rho_T>0$} is radius of largest ball inscribed in $T\in\cT_h$.

The \kk{discrete operator $L_h\from V_h \to V_h$ is defined
in terms of the bilinear form $a_L$ in~\eqref{e:a-L} via 
the relation 
$\scalar{L_h \phi_h, \psi_h}{L_2(\cD)}
=
a_L( \phi_h, \psi_h )$ 
which holds for all 
$\phi_h,\psi_h \in V_h$}.
We then consider the following SPDE
on the \kk{finite-dimensional} state space $V_h$,
\begin{align}\label{e:uh}
    L_h^\beta u_h = \white_h
    \quad
    \text{in }\cD,
\end{align}
where $\white_h$ is Gaussian white noise in
$V_h$, i.e.,
$\white_h = \sum_{j=1}^{n_h} \xi_j e_{j,h}$
for a basis $\{e_{j,h}\}_{j=1}^{n_h}$ of $V_h$ which
is orthonormal in $L_2(\cD)$ and 
$\xi_j \sim \pN(0,1)$ i.i.d.~for $j=1,\ldots,n_h$.

\kk{We note that the 
assumptions~\ref{ass:coeff-H}--\ref{ass:coeff-kappa} 
from \S\ref{subsec:fractional} 
on the functions $\mv{H}$ and $\kappa$ 
combined with the convexity of $\cD$ 
imply that the operator 
$L$ in~\eqref{e:L-div} is 
$H^2(\cD)$-regular, i.e., 
if $f\in L_2(\cD)$, then the weak solution 
$u\in V$ to $Lu=f$ satisfies $u\in H^2(\cD) \cap V$, 
see, e.g.,~\cite[Thm.~3.2.1.2]{Grisvard:2011} 
for the case of Dirichlet boundary conditions.
By combining this observation 
with the spectral asymptotics~\eqref{e:lambdaj} 
we see that the assumptions in Lemmata 
3.1 and 3.2 of \cite{bolin2017numerical} are satisfied
(since then, in their notation, $r=s=q=2$ and $\alpha=2/d$) 
and we obtain an 
error estimate for the finite element approximation
$u_h = L_h^{-\beta} \white_h$
in~\eqref{e:uh} for all $\beta\in(d/4,1)$.
Furthermore, since their derivation
requires only that $\beta>d/4$,
we can formulate this result for all such  
values of $\beta$ in the following proposition.}

\begin{proposition}\label{prop:uh}
\kk{Suppose that $\beta > d/4$ and 
that $L$ is given by \eqref{e:L-div} 
where $\mv{H}$ and $\kappa$ satisfy the 
assumptions~\ref{ass:coeff-H}--\ref{ass:coeff-kappa} 
from \S\ref{subsec:fractional}. 
Let $u$, $u_h$ be the solutions
to~\eqref{e:Lbeta} and~\eqref{e:uh}, respectively.
Then, there exists a constant $C>0$ such that, 
for sufficiently small $h$,} 
\begin{align*}
    \kk{\norm{u - u_h}{L_2(\Omega;L_2(\cD))} 
    \leq C h^{\min\{ 2\beta-d/2, \,2 \}}.}
\end{align*}
\end{proposition}

\subsection{The rational approximation}\label{subsec:rat-approx}

Proposition~\ref{prop:uh} 
shows that the mean-square error between $u$ and $u_h$ in $L_2(\cD)$
converges to zero as $h\to 0$.  
It remains to describe
how an approximation of the random field $u_h$
with values in the finite-dimensional state space $V_h$
can be constructed.

For $\beta\in\bbN$ one can
use, e.g., the iterated finite element
method presented in Appendix~\ref{app:iter-fem}
to compute $u_h$ in~\eqref{e:uh} directly.
In the following, we construct approximations of
$u_h$ if $\beta\not\in\bbN$ is a fractional exponent.
For this purpose, we aim at finding a
non-fractional equation
\begin{align}\label{e:uhr}
    \lop u_{h,m}^R = \rop \white_h
    \quad
    \text{in }\cD,
\end{align}
such that $u_{h,m}^R$ is a
good approximation of $u_h$, and where
the operator $P_{j,h} := p_j(L_h)$
is defined in terms of a polynomial
$p_{j}$ of degree \kk{$m_j\in\bbN_0$, for $j\in\{\ell,r\}$}.
Since the so-defined operators $\lop$, $\rop$
commute, this will lead to a nested
SPDE model of the form
\begin{equation}\label{e:nested-discrete}
\begin{split}
    \lop x_{h,m} &= \white_h \hspace*{1.15cm} \text{in }\cD, \\
    u^R_{h,m}    &= \rop x_{h,m} \quad \text{in }\cD,
\end{split}
\end{equation}
which facilitates efficient computations,
see \S\ref{sec:comp} and Appendix~\ref{app:iter-fem}.

Comparing the initial equation~\eqref{e:Lbeta} with
\begin{align}\label{e:ur}
    \lopex u_m^R = \ropex \white
    \quad
    \text{in }\cD,
\end{align}
where $P_j := p_j(L)$, $j\in\{\ell,r\}$,
motivates the choice $m_\ell - m_r \approx \beta$
in order to obtain a similar smoothness of
$u_m^R = (\ropex^{-1} \lopex)^{-1} \white$
and $u = L^{-\beta} \white$ in~\eqref{e:Lbeta}.
In practice, we \kk{first} choose a degree $m\in\bbN$ 
and \db{then set}
\begin{align}\label{e:betac}
    m_r := m
    \qquad
    \text{and}
    \qquad
    m_\ell := m + m_{\beta},
    \qquad
    \text{where}
    \qquad
    m_{\beta} := \max\{ 1, \floor{\beta} \}.
\end{align}
In this case, the solution $u_{m}^R$ of~\eqref{e:ur}
has the same smoothness as the solution
$v$ of the non-fractional equation
$L^{\floor{\beta}} v =  \white$, if $\beta\geq 1$, and as $v$ in
$L v =  \white$, if $\beta<1$.
Note that, for fixed $h$, the degree $m$ controls
the accuracy of the approximation~$u_{h,m}^R$.

We now turn to the problem of defining
the non-fractional operators $\lop$ and~$\rop$
in~\eqref{e:uhr}.
In order to compute $u_h$ in~\eqref{e:uh} directly,
one would have to apply the discrete fractional inverse
$L_h^{-\beta}$ to the noise term $\white_h$ on the right-hand side.
Therefore, a first idea would be to approximate
the function $x^{-\beta}$ on the spectrum of $L_h$
by a rational function $\widetilde{r}$ and
to use $\widetilde{r}(L_h) \white_h$
as an approximation of $u_h$.
This is, in essence, the approach
proposed by \cite{harizanov2016optimal}
to find optimal solvers for the problem $\mv{L}^{\beta}\mv{x} = \mv{f}$,
where $\mv{L}$ is a sparse symmetric positive definite matrix.
However, the spectra of $L$ and of $L_h$ as $h\to 0$
\kk{(considered as operators on $L_2(\cD)$)} 
are unbounded and, thus, it would be necessary
to normalize the spectrum of $L_h$ for every $h$,
since it is not feasible to construct the rational approximation $\widetilde{r}$
on an unbounded interval.
%
%
We aim at an approximation
$L_h^{-\beta} \approx p_\ell(L_h)^{-1} p_r(L_h)$,
where in practice the \kk{choice} of $p_\ell$ and $p_r$
can be made independent of $L_h$ and $h$. 
Thus, we pursue another idea.

In contrast to the 
\kk{differential operator $L$ in~\eqref{e:L-div}, 
its inverse $L^{-1}\from L_2(\cD) \to L_2(\cD)$ 
is compact and, thus}, the spectra
of $L^{-1}$ and of $L_{h}^{-1}$ are bounded subsets of
the intervals $J:=\bigl[ 0,\lambda_{1}^{-1} \bigr]$
and
$J_h := \bigl[ \lambda_{n_h,h}^{-1}, \lambda_{1,h}^{-1} \bigr] \subset J$,
respectively, 
where $\lambda_{1,h}, \lambda_{n_h,h}  > 0$
are the smallest and the largest eigenvalue of $L_h$.
\kk{This motivates 
a rational approximation $r$
of the function $f(x) := x^\beta$
on~$J$ and to deduce 
the non-fractional equation~\eqref{e:uhr}
from $u_{h,m}^R = r(L_h^{-1}) \white_h$}.


\kk{In order to achieve our envisaged choice \eqref{e:betac} 
of different polynomial degrees $m_\ell$ and~$m_r$, 
%
we decompose $f$ via $f(x) = \hat{f}(x) x^{m_\beta}$,
where $\hat{f}(x):=x^{\beta-m_\beta}$.  
We approximate $\hat{f}\approx\hat{r} := \frac{q_1}{q_2}$ on $J_h$,
where $q_1(x) := \sum_{i=0}^m c_i x^i$ and
$q_2(x) := \sum_{j=0}^{m+1} b_j x^j$ are polynomials of degree $m$ and $m+1$, respectively, 
and use 
$r(x) := \hat{r}(x) x^{m_\beta}$ as an approximation for $f$.
This construction leads
(after expanding the fraction with $x^m$)
to a rational approximation
$\frac{p_r}{p_\ell}$ of $x^{-\beta}$, 
\begin{align}\label{e:xbeta}
    x^{-\beta} =  f(x^{-1})
        \approx \hat{r}(x^{-1}) x^{-m_\beta}
        = \frac{q_1(x^{-1})}{q_2(x^{-1}) x^{m_\beta}}
        = \frac{\sum_{i=0}^{m} c_i x^{m-i}}{\sum_{j=0}^{m+1} b_j x^{m + m_\beta -j}},
\end{align}
where the polynomials 
$p_r(x) := \sum_{i=0}^{m} c_i x^{m-i}$ 
and 
$p_\ell(x) := \sum_{j=0}^{m+1} b_j x^{m + m_\beta -j}$ 
are of 
degree $m$  
and $m+m_\beta$, respectively, 
i.e., \eqref{e:betac} is satisfied}. 


The operators $\lop$, $\rop$ in \eqref{e:uhr} 
are defined accordingly,
\begin{align}\label{e:loprop}
    \lop := p_\ell(L_h) = \sum_{j=0}^{m+1} b_j L_h^{m+m_\beta-j},
    \qquad
    \rop := p_r(L_h) = \sum_{i=0}^{m} c_i L_h^{m-i}.
\end{align}
Their continuous counterparts in~\eqref{e:ur} are
$\lopex := p_\ell(L)$ and $\ropex := p_r(L)$.
We note that, for \eqref{e:betac} to hold,
any choice $m_2 \in \{0,1,\ldots, m+m_\beta\}$
would have been permissible
for the polynomial degree of $q_2$,
if $m$ is the degree of $q_1$.
The reason for setting $m_2 = m+1$ is
that this is the maximal choice which
is universally applicable
for all values of 
\kk{$\beta > d/4$}. 

In the following we refer to $u_{h,m}^R$ in \eqref{e:uhr}
with $\lop$, $\rop$ defined by \eqref{e:loprop}
as the rational \kk{SPDE} approximation of degree $m$.
We emphasize that this approximation
relies (besides the finite element discretization)
only on the rational approximation of the function~$\hat{f}$.
In particular, no information about the operator~$L$
except for a lower bound of the eigenvalues
is needed.
In the Mat\'ern case, we have $L = \kappa^2 - \Delta$
(with certain boundary conditions)
and an obvious lower bound of the eigenvalues 
is therefore given by $\kappa^2$.

\subsection{An error bound for the 
	rational approximation}\label{subsec:error}

In this \kk{subsection} we justify the 
\kk{approach proposed  
in \S\ref{subsec:discrete}--\S\ref{subsec:rat-approx}}
by providing an upper bound for the strong mean-square error
$\norm{u - u_{h,m}^R}{L_2(\Omega; L_2(\cD))}$.
Here $u$ and $u_{h,m}^R$ are the solutions of~\eqref{e:Lbeta} and \eqref{e:uhr}
and the rational approximation $u_{h,m}^R$
is constructed as described in \S\ref{subsec:rat-approx},
assuming that $\hat{r}=\hat{r}_h$ is the $L_\infty$-best
rational approximation of $\hat{f}(x) = x^{\beta - m_\beta}$
on the interval $J_h$ for each $h$.
This means that \kk{$\hat{r}_h$}
minimizes the error in the supremum norm
on $J_h$ among all
rational approximations of the chosen degrees in numerator and denominator.
How such approximations can be computed is discussed in \S\ref{subsec:rat-coeff}.

The theoretical analysis presented in Appendix~\ref{app:convergence}
results in the following theorem,
showing strong convergence of the rational approximation $u_{h,m}^R$
to the exact solution $u$.

\begin{theorem}\label{thm:strong}
	
\kk{Suppose that $\beta > d/4$ and 
that $L$ is given by \eqref{e:L-div} 
where $\mv{H}$ and $\kappa$ satisfy the 
assumptions~\ref{ass:coeff-H}--\ref{ass:coeff-kappa} 
from \S\ref{subsec:fractional}. 
Let $u$, $u_{h,m}^R$ be the solutions
to~\eqref{e:Lbeta} and~\eqref{e:uhr}, respectively.
Then, there is a constant $C>0$, 
independent of $h, m$, 
such that, 
for sufficiently small $h$,} 	
\begin{align*}
    \kk{\norm{u - u_{h,m}^R}{L_2(\Omega;L_2(\cD))}
        \leq C \Bigl( h^{\min\{ 2\beta-d/2, \, 2 \}}
              + \mathbbm{1}_{\beta\notin \mathbb{N}} h^{\min\{2(\beta-1),\, 0\}-d/2} 
              e^{-2\pi \sqrt{|\beta-m_\beta|m}} \Bigr).}
\end{align*}
\end{theorem}

\begin{remark}\label{rem:calibrate-h-m}
In order to calibrate the accuracy of the rational approximation
with the finite element error,
one can choose $m\in\bbN$ such that
\kk{$ e^{-2\pi \sqrt{|\beta - m_\beta| m}} 
\propto h^{2 \max\{\beta,\, 1\}}$. 
The strong rate of mean-square convergence is then 
$\min\{ 2\beta - d/2, \, 2 \}$.}
\end{remark}

\begin{remark}\label{rem:p-fem}
	\kk{If the functions $\mv{H}$ and $\kappa$ of the 
	operator $L$ in~\eqref{e:L-div} are smooth,  
	$\mv{H}\in C^\infty(\overline{\cD})^{d\times d}$ and 
	$\kappa\in C^\infty(\overline{\cD})$ 
	(as, e.g., in the Mat\'ern case) 
	and if the domain $\cD$ has a smooth boundary, 
	the higher-order strong mean-square convergence rate 
	$\min\{ 2\beta - d/2, \, p+1 \}$  
	can be proven for a finite element method 
	with continuous basis functions 
	which are piecewise polynomial 
	of degree at most $p\in\bbN$. 
	Thus, for $\beta>1$, 
	finite elements with 
	$p > 1$ may be meaningful}. 
\end{remark} 

\subsection{Computing the coefficients of the rational approximation}\label{subsec:rat-coeff}

As explained in \S\ref{subsec:rat-approx}, the coefficients
$\{c_i\}_{i=0}^{m}$ and $\{b_j\}_{j=0}^{m+1}$ needed for defining the operators
$\lop$, $\rop$ in~\eqref{e:loprop}
are obtained from a rational approximation $\hat{r} = \hat{r}_h$
of $\hat{f}(x) = x^{\beta - m_\beta}$ on $J_h$.
For each $h$, this approximation can, e.g., be computed
with the second Remez algorithm \citep{remez1934determination},
which generates the coefficients
of the $L_\infty$-best approximation.
The error analysis for the resulting
approximation $u_{h,m}^R$ in~\eqref{e:uhr} was performed 
in \S\ref{subsec:error}.
Despite the theoretical benefit of generating the
$L_\infty$-best approximation, the Remez algorithm
is often unstable in computations
and, therefore, we use a different method in our simulations.
However, versions of the Remez scheme were used, e.g., by \cite{harizanov2016optimal}.

A simpler and computationally more stable way
of choosing the rational approximation is, for instance,
the Clenshaw--Lord Chebyshev--Pad\'e algorithm \citep{baker1996pade}.
%
\db{To \kk{further} improve the stability of the method,
we will rescale the operator \kk{$L$} so that 
\kk{its eigenvalues are bounded from below by one},
which for the Mat\'ern case corresponds 
to reformulating the SPDE~\eqref{e:statmodel} as 
$(\mathrm{Id} - \kappa^{-2}\Delta)^{\beta} (\widetilde{\tau} u) = \white$ 
\kk{and using $L = \mathrm{Id} - \kappa^{-2}\Delta$, 
where $\mathrm{Id}$ denotes the identity on $L_2(\cD)$ 
and $\widetilde{\tau} := \kappa^{2\beta} \tau$}. }

In order to avoid computing  
a different rational approximation $\hat{r}$
for each finite element mesh \kk{width} $h$,
in practice we compute the approximation $\hat{r}$
only once on the interval $J_{*} := [\delta,1]$,
where $\delta\in(0,1)$ should ideally be chosen
such that $J_h \subset J_{*}$ for all considered
mesh sizes $h$.
For the numerical experiments later,
we will use $\delta = 10^{-(5+m)/2}$
when computing rational approximations of order $m$,
which gives acceptable results for all values of $\beta$.
As an example, the coefficients
computed with the Clenshaw--Lord Chebyshev--Pad\'e algorithm
on $J_{*}$
for the case of exponential covariance on $\bbR^2$
are shown in Table~\ref{tab:coeffs}.

\begin{table}
\centering
\begin{tabular}{lccccccccc}
\toprule
$m$ & $b_0$ & $c_0$  & $b_1$ & $c_1$ 	& $b_2$ & $c_2$ & $b_3$ & $c_3$ & $b_4$ \\
\cmidrule(r){2-10}
1 & 1.69e-2 & 7.69e-2 & 8.06e-1 & 1 		 & 2.57e-1 & 	& 		&	             & 		 \\
2 & 8.08e-4 & 5.30e-3 & 1.98e-1 & 4.05e-1 & 1.07  & 1	        & 1.41e-1 & 	& \\
3 & 3.72e-5 & 3.27e-4 & 3.03e-2 & 8.57e-2 & 6.84e-1 & 1.00 & 1.28       &	 1	& 9.17e-2 \\
\bottomrule
\end{tabular}
\vspace{0.2cm}
\caption{Coefficients of the rational approximation for $\beta = 3/4$
(exponential covariance on~$\bbR^2$) for $m=1,2,3$,
normalized so that $c_{m} = 1$.}
\label{tab:coeffs}
\end{table}

\section{Computational aspects of the rational approximation}
\label{sec:comp}
In the non-fractional case,
the sparsity of the precision matrix for the weights~$\mv{u}$ in \eqref{e:basisexp}
facilitates
fast computation of samples, likelihoods,
and other quantities of interest
for statistical inference.
The purpose of this section
is to show that the rational SPDE approximation
proposed in \S\ref{sec:rational}
preserves these good computational properties.

%
%
\kk{The representation~\eqref{e:nested-discrete} shows that $u_{h,m}^R$ 
can be seen as a Markov random field $x_{h,m}$, 
transformed by the operator $P_{r,h}$.
%
Solving this latent model 
as explained in Appendix~\ref{app:iter-fem}, 
yields an approximation of the form \eqref{e:basisexp}},
where $\mv{\Sigma}_{\mv{u}} = \rmat\mv{Q}^{-1}\rmat^{\trsp}$.
Here $\lmat, \rmat \in~\bbR^{n_h\times n_h}$ correspond to
the discrete operators~$\lop$ and $\rop$ in~\eqref{e:loprop}, respectively.
The matrix $\mv{Q} := \lmat^{\trsp} \mv{C}^{-1} \lmat$
is sparse if the mass matrix $\mv{C}$
with respect to the finite element basis $\{\varphi_j\}_{j=1}^{n_h}$
is replaced by the diagonal lumped mass matrix $\widetilde{\mv{C}}$,
see Appendix~\ref{app:iter-fem}.
By defining $\mv{x} \sim \pN(\mv{0},\mv{Q}^{-1})$,
we have $\mv{u} = \rmat\mv{x}$,
which is a transformed Gaussian Markov random field (GMRF).
Choosing $\mv{x}$ as a latent variable instead of $\mv{u}$
thus enables us to use all computational methods,
which are available for GMRFs \citep[see][]{rue05},
also for the rational SPDE approximation.

As an illustration, we consider the following hierarchical model,
with a latent field $u$ which is a rational approximation of \eqref{e:Lbeta},
\begin{equation}\label{e:model-yi}
\begin{split}
y_i &= u(\svec_i) + \varepsilon_i, \quad i=1,\ldots, N, \\
\lopex u &= \ropex \white \qquad\quad\, \,\, \text{in }\cD,
\end{split}
\end{equation}
where $u$ is observed under
i.i.d.~Gaussian measurement noise
$\varepsilon_i \sim \pN(0,\sigma^2)$.
Given that one can treat this case,
one can easily adapt the method to be used for inference in combination with MCMC or INLA \citep{rue09}
for models with more sophisticated likelihoods.

Defining the matrix $\mv{A}$ with elements
$A_{ij}= \varphi_j(\svec_i)$
and the vector $\mv{y} = (y_1,\ldots, y_N)^{\trsp}$
gives us the discretized model
\begin{equation}\label{e:model-ygivenx}
\begin{split}
    \mv{y}|\mv{x} &\sim \pN(\mv{A} \rmat \mv{x},\sigma^2\mv{I}), \\
    \mv{x}        &\sim \pN(\mv{0}, \mv{Q}^{-1}).
\end{split}
\end{equation}
In this way, the problem has been reduced
to a standard latent GMRF model
and a sparse Cholesky factorization
of $\mv{Q}$ can be used for sampling $\mv{x}$
from \db{$\pN(\mv{0}, \mv{Q}^{-1})$} 
as well as to evaluate \db{its log-density}~$\log\pi_x(\mv{x})$.
Samples of $\mv{u}$ can then be obtained
from samples of $\mv{x}$
via $\mv{u} = \rmat \mv{x}$.
For evaluating \db{the log-density of $\mv{u}$}, $\log\pi_u(\mv{u})$,
the relation $\log\pi_u(\mv{u}) = \log\pi_x(\rmat^{-1}\mv{u})$
can be exploited.
Furthermore, the posterior distribution of $\mv{x}$ is 
\kk{given by} 
$\mv{x}|\mv{y} \sim \pN\bigl( \mv{\mu}_{\mv{x}|\mv{y}}, \mv{Q}_{\mv{x}|\mv{y}}^{-1} \bigr)$,
where $\mv{\mu}_{\mv{x}|\mv{y}} = \sigma^{-2} \mv{Q}_{\mv{x}|\mv{y}}^{-1} 
	\rmat^{\trsp} \mv{A}^{\trsp} \mv{y}$
and
$\mv{Q}_{\mv{x}|\mv{y}} 
	= \mv{Q} + \sigma^{-2} \rmat^{\trsp} \mv{A}^{\trsp} \mv{A} \rmat$
is a sparse matrix.
Thus, simulations from \db{the distribution of $\mv{x}|\mv{y}$}, 
and evaluations
of \db{the corresponding log-density} 
$\log\pi_{x|y}(\mv{x})$, can be performed efficiently
via a sparse Cholesky factorization of $\mv{Q}_{\mv{x}|\mv{y}}$.
Finally, the marginal data log-likelihood is proportional to
\begin{align*}
    \log| \lmat | - \frac{1}{2} \log| \mv{Q}_{\mv{x}|\mv{y}} | 
    - N \log\sigma
    - \frac{1}{2} \left( \mv{\mu}_{\mv{x}|\mv{y}}^{\trsp} \mv{Q} \, \mv{\mu}_{\mv{x}|\mv{y}} 
    + \sigma^{-2} \left\|\mv{y} - \mv{A} \rmat \mv{\mu}_{\mv{x}|\mv{y}} \right\|^2\right).
\end{align*}

We therefore conclude that all
computations needed for statistical inference
can be facilitated by sparse Cholesky factorizations of
$\lmat$ and $\mv{Q}_{\mv{x}|\mv{y}}$.

\begin{remark}\label{rem:compcosts}
From the specific form of the matrices $\lmat$ and $\rmat$
addressed in Appendix~\ref{app:iter-fem},
we can infer that the number of non-zero elements in $\mv{Q}_{\mv{x}|\mv{y}}$
for a rational SPDE approximation of degree $m$
will be the same as the number of non-zero elements in $\mv{Q}_{\mv{x}|\mv{y}}$
for the standard (non-fractional) SPDE approach with $\beta = m+m_\beta$.
Thus, also the computational cost
will be comparable for these two cases.
\end{remark}

\begin{remark}
The matrix $\mv{Q}_{\mv{x}|\mv{y}}$ can be
ill-conditioned for $m>1$ if a FEM approximation
with piecewise \kk{linear} basis functions is used.
The numerical stability for large values of $m$ can likely
be improved by increasing the polynomial degree of the FEM basis functions, see also Remark~\ref{rem:p-fem}.
\end{remark}

\section{Numerical experiments}\label{sec:numerics}

\subsection{The Mat\'ern covariance on $\bbR^2$}\label{subsec:numerics-matern}
As a first test, we investigate the performance of the rational SPDE approach
for Gaussian Mat\'ern fields, without including the finite element discretization in space.

The spectral density $S$ of the solution to \eqref{e:statmodel} on $\bbR^2$
is given by \eqref{eq:maternSpec},
whereas the spectral density for the non-discretized rational SPDE approximation $u_m^R$ in~\eqref{e:ur} is
\begin{align}\label{e:S-rat}
    S_R(\mv{k}) \propto \kappa^{4\beta}\left(\frac{\sum_{i=1}^m c_i(1+\kappa^{-2}\|\mv{k}\|^2)^{m-i}}{\sum_{j=1}^{m+1} b_j(1+\kappa^{-2}\|\mv{k}\|^2)^{m+m_\beta-j}}\right)^2.
\end{align}
We compute the coefficients as described in \S\ref{subsec:rat-coeff}.
To this end, we apply an implementation of the
Clenshaw--Lord Chebyshev--Pad\'e algorithm
provided by the Matlab package Chebfun \citep{driscoll2014chebfun}.
By performing a partial fraction decomposition of \eqref{e:S-rat},
expanding the square, transforming to polar coordinates, and using the equality
\begin{align*}
    \int_0^{\infty} \frac{\omega J_0(\omega h)}{(\omega^2+a^2)(\omega^2+b^2)} \,\rd\omega
    = \frac{1}{(b^2-a^2)}(K_0(ah)-K_0(bh)),
\end{align*}
we are able to compute the corresponding covariance function $C_R(h)$ analytically.
Here, $J_0$ is a Bessel function of the first kind
and $K_0$ is a modified Bessel function of the second kind.
To measure the accuracy of the approximation,
we compare $C_R(h)$ to the true Mat\'ern covariance function $C(h)$
for different values of $\nu$, where $\kappa=\sqrt{8\nu}$ is chosen such that
the practical correlation range 
\kk{$r=\sqrt{8\nu}/\kappa$}
equals one in all cases.

\begin{figure}[t]
\begin{center}
\includegraphics[width=0.49\linewidth]{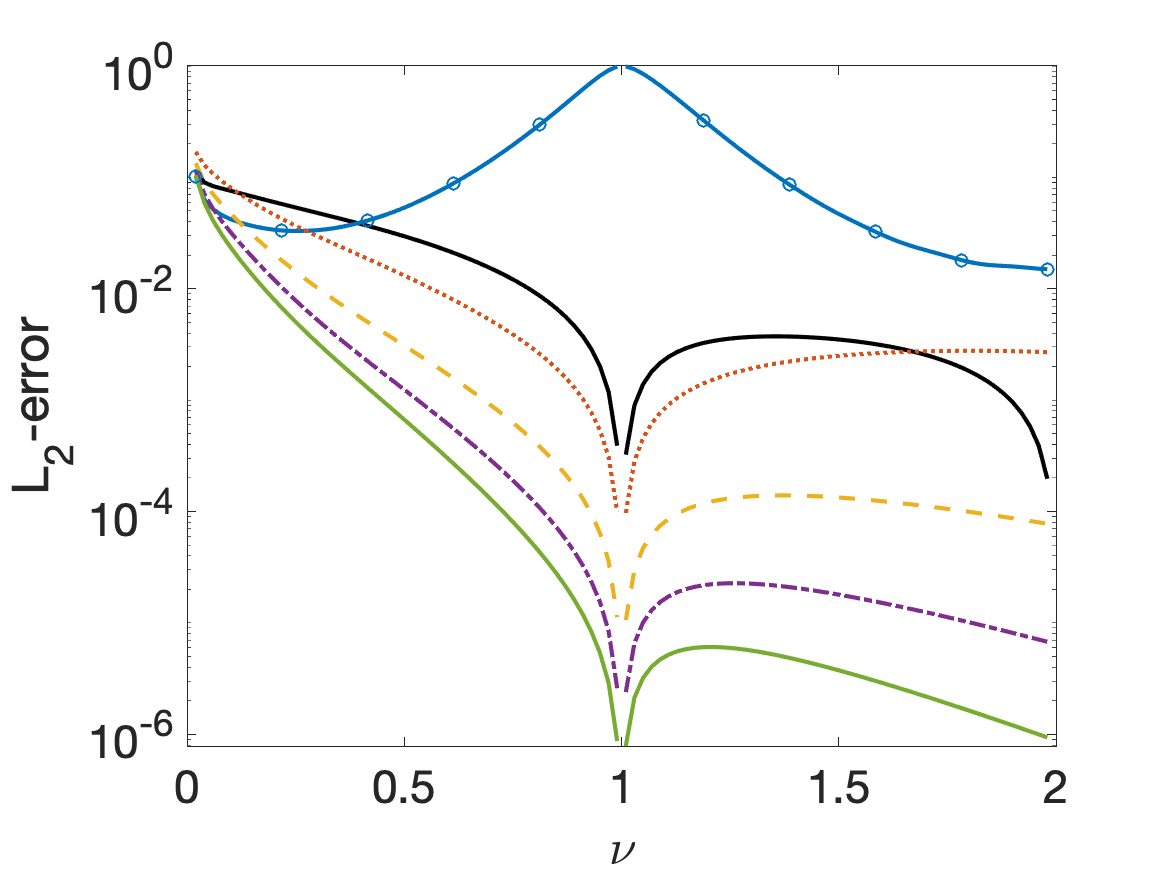}
\includegraphics[width=0.49\linewidth]{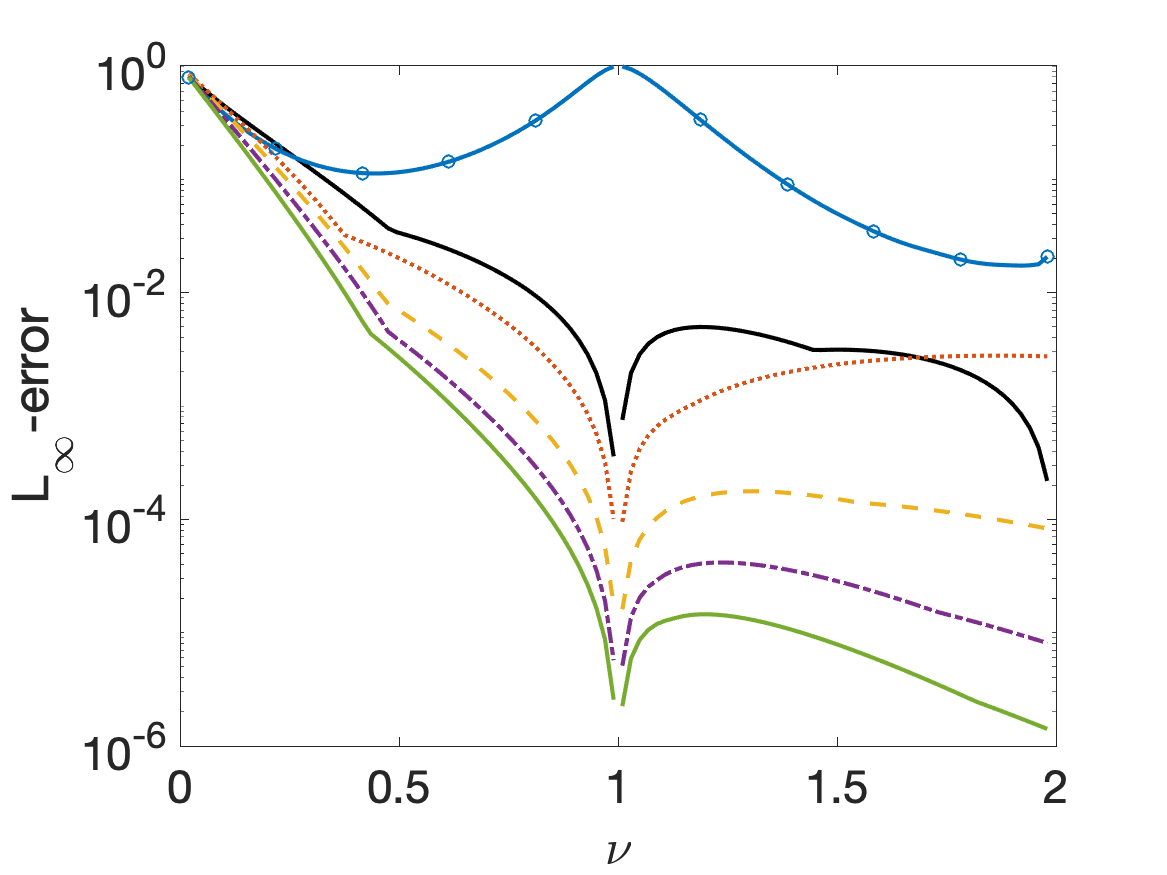}
\includegraphics[width=0.4\linewidth]{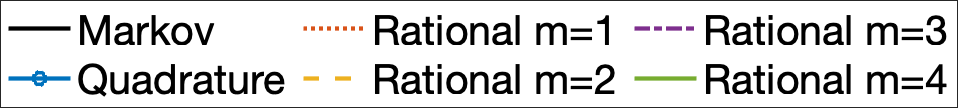}
\end{center}
\vspace{-0.2cm}
\caption{\label{fig:rational_errors}
The $L_2$- and $L_{\infty}$-errors of the covariance functions for different values of $\nu$ for the different approximation methods. When $\nu=1$, all methods are exact.
}
\end{figure}

To put the accuracy of the rational approximation in context,
the Markov approximation by \cite{lindgren11} and the quadrature method by \cite{bolin2017numerical} are also shown.
For the quadrature method, $K=12$ quadrature nodes are used,
which results in an approximation with the same computational cost
as a rational approximation of degree $m=11$, see Appendix~\ref{subsec:rat-comparequad}.
Figure~\ref{fig:rational_errors} shows the normalized error in the $L_2$-norm
and the error with respect to $L_\infty$-norm
for different values of $\nu$,
both with respect to the
interval~$[0,2]$ of length twice the practical correlation range, i.e.,
\begin{align*}
    \left( \frac{\int_0^2 (C(h)-C_a(h))^2 \, \rd h}{\int_0^2 C(h)^2 \, \rd h} \right)^{1/2}
    \quad
    \text{and}
    \quad
    \sup_{h\in[0,2]} |C(h) - C_a(h)|.
\end{align*}
Here,
$C_a$ is the covariance function
obtained by the respective approximation method.

Already for $m=3$, the rational approximation performs better than both
the Markov approximation and the quadrature approximation for all values of $\nu$.
It also decreases the error for the case of an exponential covariance by several orders of magnitude.

All methods are exact when $\nu=1$, 
since this is the non-fractional case. 
The Markov and rational methods 
show errors decreasing to zero as $\nu=1$, 
whereas the error of the quadrature method has  
\kk{a singularity at $\nu=1$}. 
The performance of the quadrature method can be improved (although not the behavior near $\nu=1$)
by increasing the number of quadrature nodes, see Appendix~\ref{subsec:rat-comparequad}.
This is reasonable if the method is needed only for sampling from the model,
but implementing this method for statistical applications,
which require kriging or likelihood evaluations, is not feasible
since the computational costs then are comparable
to the standard SPDE approach with $\beta=K$. 

Finally, it should be noted that the Markov method 
also is exact at $\nu=2$ ($\beta = 1.5$) 
since the spectrum of the process then is the reciprocal of a polynomial. 
The rational and quadrature methods 
cannot \kk{exploit} this fact, since 
\kk{these approximations 
are based on the corresponding differential operator 
instead of the spectral density}. 
This is the prize that has to be paid 
in order to \kk{formulate} a method which works 
not only for the stationary Mat\'ern fields 
but also for non-stationary and non-Gaussian models.

\subsection{Computational cost and the finite element error}
From the study in the previous subsection,
we infer that the rational SPDE approach performs well
for Mat\'ern fields with arbitrary smoothness.
However, as for the standard SPDE approach, we need to discretize
the problem in order to be able to use the method in practice, e.g., for inference.
This induces an additional error source,
which means that one should balance the two errors
by choosing the degree $m$ of the rational approximation
appropriately with respect to the FEM error.
\kk{A calibration 
based on the theoretical results has been suggested in 
Remark~\ref{rem:calibrate-h-m}}.
In this section we address this issue in practice
and investigate the computational cost of the rational SPDE approximation.

As a test case, we compute approximations
of a Gaussian Mat\'ern field
with unit variance and
practical correlation range $r=0.1$
on the unit square in $\bbR^2$.
We assume homogeneous Neumann boundary conditions
for the Mat\'ern operator $\kappa^2-\Delta$ in \eqref{e:statmodel}.
For the discretization, we use
a FEM with a nodal basis of continuous piecewise \kk{linear} functions
with respect to a mesh induced by a Delaunay triangulation
of a regular lattice on the domain, with a total of $n_h$ nodes.
We consider three different meshes
with $n_h = 57^2, 85^2, 115^2$,
which corresponds to $h \approx r/4, r/6, r/8$.

In order to measure the accuracy,
we compute the covariances between the
midpoint of the domain $\tilde{\mv{s}}_{*}$
and all other nodes in the lattice $\{\tilde{\mv{s}}_{j}\}_{j=1}^{n_h}$
for the Mat\'ern field
and the rational SPDE approximations
and calculate the error similarly
to the $L_2$-error in \S\ref{subsec:numerics-matern},
\begin{align*}
    \left( \frac{ \sum_{j=1}^{n_h} ( C(\| \tilde{\mv{s}}_{*}  - \tilde{\mv{s}}_{j}\|) - \Sigma^{\mv{u}}_{j,*} )^2 }{ \sum_{j=1}^{n_h} C(\| \tilde{\mv{s}}_{*}  - \tilde{\mv{s}}_{j}\|)^2 } \right)^{1/2}, 
\end{align*}
where $\mv{\Sigma}^{\mv{u}} = \rmat\lmat^{-1} \mv{C} \lmat^{-\trsp} \rmat^{\trsp}$
is the covariance matrix of $\mv{u}$, see Appendix~\ref{app:iter-fem}.
As a consequence of imposing boundary conditions,
the error of the covariance is larger close
to the boundary of the domain.
However, we compare this error to the error of
the non-fractional SPDE approach, which has the same boundary effects.
%
%
As measures of the computational cost,
we consider the time it takes to sample $\mv{u}$
and to evaluate $\log|\mv{Q}_{\mv{x}|\mv{y}}|$
for the model \eqref{e:model-ygivenx} with $\sigma=1$,
when~$\mv{y}$ is a vector of noisy observations
of the latent field at $1000$ locations, drawn at random in the domain
(a similar computation time is needed to evaluate $\mv{\mu}_{\mv{x}|\mv{y}}$).

\begin{table}[t]
\centering
\begin{tabular}[b]{llcccccc}
\toprule
&& \multicolumn{3}{c}{Rational SPDE approximation} & \multicolumn{3}{c}{\db{Standard SPDE approach}}\\
\cmidrule(r){3-5}   \cmidrule(r){6-8}
$n$			&			& $m=1$ & $m=2$ & $m=3$ & $\beta = 2$ & $\beta = 3$ & $\beta = 4$\\
\cmidrule(r){1-8}
\multirow{2}{*}{$57^2$} 	& Error	& 1.849 &  1.339 &   1.415 &   2.259  &  2.173 &   2.147\\
				& Time 		& 1.5  (3.2)  &  1.8  (5.2) &   2.7   (8.7)  &  1.7  (2.6) &  1.7 (3.9) &  2.2  (6.3)\\
\cmidrule(r){2-8}
\multirow{2}{*}{$85^2$} & Error		& 1.720  &  0.757  &  0.807  &  0.953  &  0.928   & 0.921\\
				& Time 		& 3.1  (8.4)  &  5.0 (14)  &  7.6 (25)  &  3.0 (8.2)  &  5.8  (13) &   7.9 (22)\\
\cmidrule(r){2-8}
\multirow{2}{*}{$115^2$} & Error	& 1.559  &  0.526  &  0.501  &  0.509  &  0.498  &  0.494\\
				& Time 		& 7.6 (22) &  11 (34)  & 18 (57)  &  6.3 (18)  & 11 (35) &  18  (53)\\

\bottomrule
\end{tabular}
\vspace{0.2cm}
\caption{\label{tab:femerrors}Covariance errors ($\times 100$) and 
	computing times in seconds ($\times 100$)
	for sampling from the rational SPDE approximation $\mv{u}$ \db{(with $\beta=3/4$)}
	and, in  parentheses, 
	for evaluating $\log|\mv{Q}_{\mv{x}|\mv{y}}|$.
	These values are also given
	for the standard SPDE approach with $\beta=2,3,4$.}
\end{table}

The results for rational SPDE approximations of different degrees
for the case $\beta = 3/4$ (exponential covariance)
are shown in Table~\ref{tab:femerrors}.
Furthermore,
we perform the same experiment
when the standard (non-fractional) SPDE approach is used for $\beta = 2,3,4$.
As previously mentioned in Remark~\ref{rem:compcosts},
the computational cost of the rational SPDE approximation of degree~$m$
should be comparable to the standard SPDE approach with $\beta= m+1$.
Table~\ref{tab:femerrors} validates this claim.
One can also note that the errors of the rational SPDE approximations
are similar to those of the standard SPDE approach,
and that the reduction in error
when increasing from $m=2$ to $m=3$
is small for all cases,
indicating that the error induced by the rational approximation
is small compared to the FEM error, even for a low degree~$m$.
This is also the reason for why, in particular in the pre-asymptotic region,
one can in practice choose the degree $m$ smaller than the value suggested in
Remark~\ref{rem:calibrate-h-m},
which gives $m\approx 6, 7, 8$ for $\beta=3/4$ and the three considered finite element meshes.

\section{Likelihood-based inference of Mat\'{e}rn parameters}\label{sec:inference}
The  computationally efficient evaluation of the likelihood
of the rational SPDE approximation facilitates
likelihood-based inference for all parameters of the Mat\'ern model,
including~$\nu$ which until now had to be fixed
when using the SPDE approach.
In this section we first discuss the identifiability of the model parameters and then investigate the accuracy of this approach
within the scope of a simulation study.

\subsection{Parameter identifiability}\label{subsec:measure}
\kk{A common reason for fixing the smoothness 
in Gaussian Mat\'ern models is the result by \citet{Zhang2004} 
which shows that all three Mat\'ern parameters 
cannot be estimated consistently under infill asymptotics. 
More precisely, for a fixed smoothness parameter $\nu$, 
one cannot estimate both the variance of the field, 
$\phi^2$, and the scale parameter, $\kappa$, consistently. 
However, the quantity  $\phi^2\kappa^{2\nu}$  can be estimated consistently. 
The derivation of this result relies on the 
equivalence of Gaussian measures corresponding 
to Mat\'ern fields \citep[Theorem~2]{Zhang2004}. 
%
%
The following theorem provides the analogous result 
for the Gaussian measures induced by the class 
of random fields specified via \eqref{e:statmodel} 
on a bounded domain}. \kk{Its proof can be found in Appendix~\ref{sec:measureproof}. 

\begin{theorem}\label{thm:measure}
Let $\cD\subset\bbR^d$, $d\in\{1,2,3\}$,  
be bounded, open and connected.  	
For $i\in\{1,2\}$, 
let $\beta_i > d/4$, $\kappa_i, \tau_i > 0$, and 
consider the centered Gaussian measure 
$\mu_i:=\pN(0,\mathcal{Q}_i^{-1})$ 
on $L_2(\cD)$ 
with precision operator  
$\cQ_i := \tau_i^{2}L_i^{2\beta_i}$, 
where the operators $L_i:=\kappa_i^2 - \Delta$, $i\in\{1,2\}$, 
are augmented with the same homogeneous Neumann or Dirichlet boundary 
conditions. 
Then, $\mu_1$ and $\mu_2$ are equivalent 
if and only if $\beta_1=\beta_2$ and $\tau_1 = \tau_2$. 
\end{theorem}

Note that, for $\cD:=\mathbb{R}^d$, 
the parameter $\tau$ 
is related to the variance 
of the Gaussian random field 
via $\phi^2 = \Gamma(\nu)(\tau^2\Gamma(2\beta)(4\pi)^{d/2}\kappa^{2\nu})^{-1}$. 
Thus, $\tau^{-2}\propto \phi^2\kappa^{2\nu}$, 
which means that Theorem~\ref{thm:measure} 
is in accordance with the result by \citet{Zhang2004}.
Since the Gaussian measures induced by the operators 
$L_1= \tau(\kappa_1+\Delta)^{\beta}$ 
and $L_2= \tau(\kappa_2+\Delta)^{\beta}$ are equivalent, 
we will not be able to consistently estimate $\kappa$ 
under infill asymptotics. 
Yet, Theorem~\ref{thm:measure} 
suggests that it is possible  
to estimate $\tau$ and $\beta$ consistently. 
In fact, with Theorem~\ref{thm:measure} available, 
it is straightforward to show 
that $\tau$ can be estimated consistently for a fixed $\nu$ 
by exploiting the same arguments 
as in the proof of \citep[Theorem 3]{Zhang2004}. 
However, it is beyond the scope 
of this article 
to show that both $\nu$ and $\tau$ 
can be estimated consistently
which would also extend the results  
by \citet{Zhang2004}}.

\subsection{Simulation study}\label{subsec:estimationstudy}
\db{To numerically investigate the accuracy of likelihood-based parameter estimation using the rational SPDE approach, }
we again assume homogeneous Neumann boundary conditions for
the Mat\'ern operator in \eqref{e:statmodel} and consider
the standard latent model \eqref{e:model-yi}
from \S\ref{sec:comp}.
We take the unit square as the domain of interest,
set $\sigma^2=0.1$, $\nu=0.5$ and choose $\kappa$ and $\tau$
so that the latent field has variance $\phi^2=1$ and
practical correlation range~\kk{$r=0.2$}.
For the FEM, we take a mesh based on a regular lattice on the domain,
extended by twice the correlation range in each direction
to reduce boundary effects, yielding a mesh with approximately $3500$ nodes.

As a first test case, we use simulated data from the discretized model.
We simulate $50$ replicates of the latent field,
each with corresponding noisy observations at $1000$ measurement locations
drawn at random in the domain.
This results in a total of $50000$ observations,
which we use to estimate the parameters of the model.
We draw initial values for the parameters at random
and then numerically optimize the likelihood of the model
with the function \texttt{fminunc} in Matlab.
This procedure is repeated $100$ times, each time with a new simulated data set.

As a second test case, we repeat the simulation study,
but this time we simulate the
data from a Gaussian Mat\'ern field with an exponential covariance function
instead of from the discretized model.
For the estimation,
we compute the rational SPDE approximation for the same
finite element mesh as in the first test case.
To investigate the effect of the mesh resolution on the parameter estimates,
we also estimate the parameters using a uniformly refined mesh 
with twice as many nodes. 
The average computation time for 
evaluating the likelihood is approximately $0.16s$
for the coarse mesh and $0.4s$ for the fine mesh.
This computation time is affine with respect to
the number of replicates, and with only one replicate
it is $0.09s$ for the coarse mesh and $0.2s$ for the fine mesh.
\begin{table}
\centering
\begin{tabular}{ccccc}
\toprule
& 			& Rational samples & \multicolumn{2}{c}{Mat\'ern samples}\\
\cmidrule(r){3-3}   \cmidrule(r){4-5}
			& Truth	& Estimate		& Coarse mesh		& Fine mesh \\
$\kappa$		& 10 		& 10.026 (0.5661) 	& 10.966 (1.8060) 	& 10.864 (0.4414)\\
$\phi^2$		& 1.0 	& 1.0014 (0.0228) 	& 1.1089 (0.6155) 	& 0.9743 (0.0210)\\
$\sigma^2$	& 0.1 	& 0.1001 (0.0009) 	& 0.3016 (0.0036) 	& 0.2320 (0.0044)\\
$\nu$ 		& 0.5 	& 0.5011 (0.0168)	& 0.5554 (0.0991) 	& 0.5462 (0.0138)\\
\bottomrule
\end{tabular}
\vspace{0.2cm}
\caption{\label{tab:estimation}Results of the parameter estimation.
For each parameter estimate, the mean of 100 different estimates is shown, with the corresponding standard deviation in parentheses.}
\end{table}

The results of the parameter estimation can be seen in Table~\ref{tab:estimation},
where the true parameter values are shown
together with the mean and standard deviations
of the $100$ estimates for each case.
Notably, we are able to estimate all parameters
accurately in the first case.
For the second case, the finite element discretization
seems to induce a small bias, especially for the nugget estimate ($\sigma^2$)
that depends on the resolution of the mesh.
The bias in the nugget estimate is not surprising
since the increased nugget compensates
for the FEM error.
The bias could be decreased
by choosing the mesh more carefully,
also taking the measurement locations into account.
In practice, however, this bias will
not be of great importance,
since the optimal nugget for the discretized model
should be used.

It should be noted that there are several other methods
for decreasing the computational cost of likelihood-based inference for stationary Mat\'ern models.
The major advantage of the rational SPDE approach is
that it is directly applicable to more complicated
non-stationary models, which we will use in the next section when analyzing real data.

\section{Application}\label{sec:application}

In this section we illustrate
for the example of a climate reanalysis data set
how the rational SPDE approach can be used for spatial modeling.

Climate reanalysis data is generated
by combining a climate model with observations
in order to obtain a description of the recent climate.
We use reanalysis data generated with
the Experimental Climate Prediction Center Regional Spectral Model (ECPC-RSM)
which was originally prepared for
the North American Regional Climate Change Assessment Program (NARCCAP)
by means of NCEP/DOE Reanalysis \citep{mearns2007,mearns2009regional}.
As variable we consider average summer precipitation
over the conterminous U.S.\ for a 26 year period from 1979 to 2004.
The average value for each grid cell and year is computed
as the average of the corresponding daily values
for the days in June, July, and August.
\db{In order to \kk{obtain} data 
which can be modelled by a Gaussian distribution, 
we follow \cite{genton2015cross} 
and transform the data by taking the cube root. 
We then subtract the mean over the 26 years 
from each grid cell so that 
we can assume that the data has zero mean 
and focus on the correlation structure of the residuals.}
The resulting residuals for the year 1979 
are shown in Figure~\ref{fig:application_data}.

\begin{figure}[t]
\begin{center}
\includegraphics[width=0.75\linewidth]{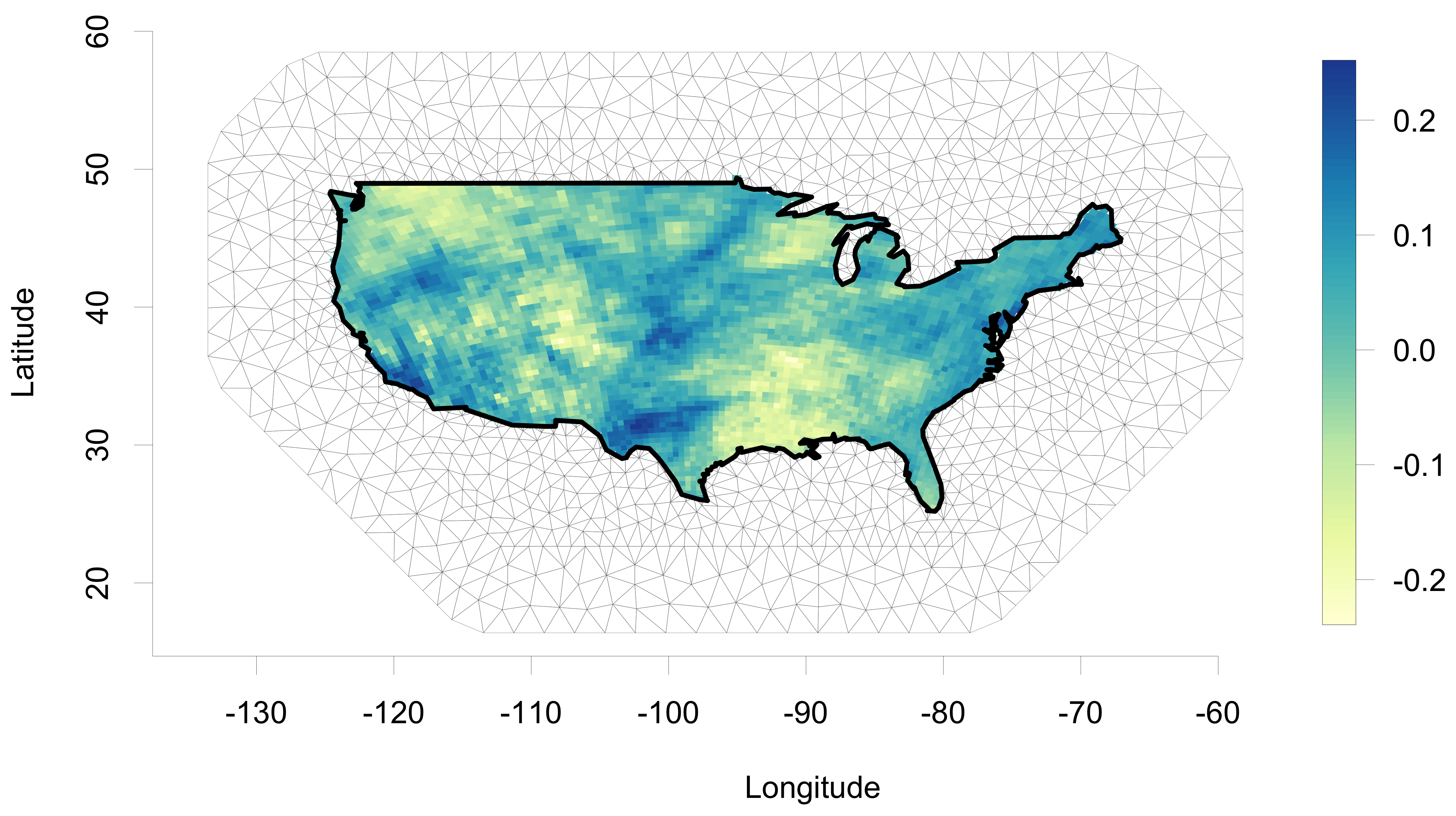}
\end{center}
\vspace{-0.2cm}
\caption{\label{fig:application_data}Average summer precipitation residuals (in cm) 
	for 1979 and the \kk{FEM mesh}. 
	}
\end{figure}


The 4112 observed residuals for each year are modeled as independent realizations
of a zero-mean Gaussian random field with a nugget effect.
That is, the measurement $Y_{ij}$ at spatial location $\mv{s}_i$
for year $j$ is modeled as
    $Y_{ij} = u_j(\mv{s}_i) + \varepsilon_{ij}$,
where $\varepsilon_{ij} \sim \pN(0,\sigma^2)$ are independent,
and \kk{$\{u_j(\mv{s})\}_j$} are independent realizations of a
zero-mean Gaussian random field $u(\mv{s})$.
\db{The analysis of \cite{genton2015cross} 
revealed that an exponential covariance model 
is suitable for a subset of this data set. 
Because of this, a natural first choice is 
to use a stationary Mat\'ern model \eqref{e:statmodel}, 
either with $\beta = 0.75$  (exponential covariance) 
or with a general $\beta$ which we estimate from the data. 
However, since we have data for a larger spatial region 
than \cite{genton2015cross}, one would suspect 
that a non-stationary model for $u(\mv{s})$ might be needed. 
The standard non-stationary model for the SPDE approach, 
as \kk{first} suggested by \cite{lindgren11} and used in many applications since then, is 
\begin{equation}\label{eq:nonstat_model}
    (\kappa(\mv{s})^2 - \Delta)^{\beta} \, (\tau(\mv{s})u(\mv{s})) = \white(\mv{s}),
    \quad
    \mv{s}\in\cD,
\end{equation}
where \kk{$\beta=1$ is fixed}. 
Until now, it has not been possible to use 
\kk{the model} \eqref{eq:nonstat_model} 
with fractional smoothness.  
\kk{Therefore, our main question is now: 
What is more important for this 
data---the fractional smoothness 
$\beta$ or the non-stationary parameters?} 
We thus consider four different SPDE models for $u(\mv{s})$.
\kk{Two of them are non-fractional models, 
where $\beta=1$ is fixed, 
and for the other two (fractional) 
models, we estimate the fractional order 
$\beta$ jointly with the other parameters
from the data. 
For both cases, 
we consider stationary Mat\'ern 
and non-stationary models, 
where the latter are formulated via \eqref{eq:nonstat_model} 
with}}
%
%
\begin{align*}
\log\kappa(\mv{s}) 
= 
\kappa_0 
+ \kappa_{\rm a} \psi_{\rm a}(\mv{s}) 
+ \sum_{i,j=1}^{2}\sum_{k,\ell=1}^{2} 
	\kappa_{ij}^{k\ell} \, 
	\psi_{i}^{k}(\widetilde{s}_1) \, 
	\psi_{j}^{\ell}(\widetilde{s}_2), 
\end{align*}
and the same model is used for $\tau(\mv{s})$.  
Here, 
$\psi_{j}^{1}(\widetilde{s}) := \sin(j\pi\widetilde{s})$,  
$\psi_{j}^{2}(\widetilde{s}) := \cos(j\pi\widetilde{s})$, 
$\psi_{\rm a}(\mv{s})$ is the altitude at location $\mv{s}$,  
and $\widetilde{\mv{s}} = (\widetilde{s}_1,\widetilde{s}_2)$ 
denotes the spatial coordinate after rescaling 
so that the observational domain is mapped to the unit square. 
Thus, $\log\kappa(\mv{s})$ and $\log\tau(\mv{s})$ are modelled 
by the altitude covariate and 16 additional Fourier basis functions 
to capture large-scale trends in the parameters. 
The altitude covariate and the eight Fourier basis functions 
$\left\{\psi_1^{k}(\widetilde{s}_1) \psi_j^{\ell}(\widetilde{s}_2) 
: j,k,\ell=1,2\right\}$
are shown 
in Figure~\ref{fig:basis_functions}.

\begin{figure}[t]
\begin{center}
\begin{minipage}{0.25\linewidth}
\includegraphics[width=\linewidth]{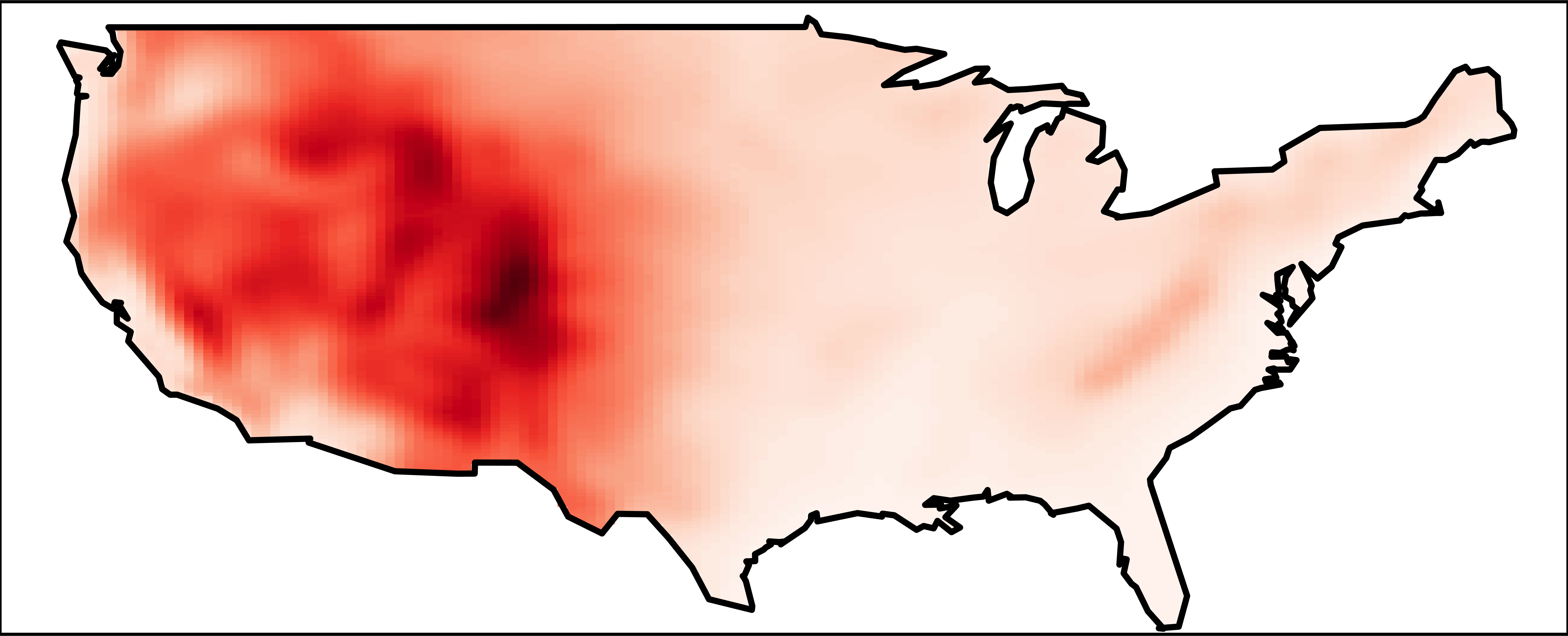}
\includegraphics[width=\linewidth]{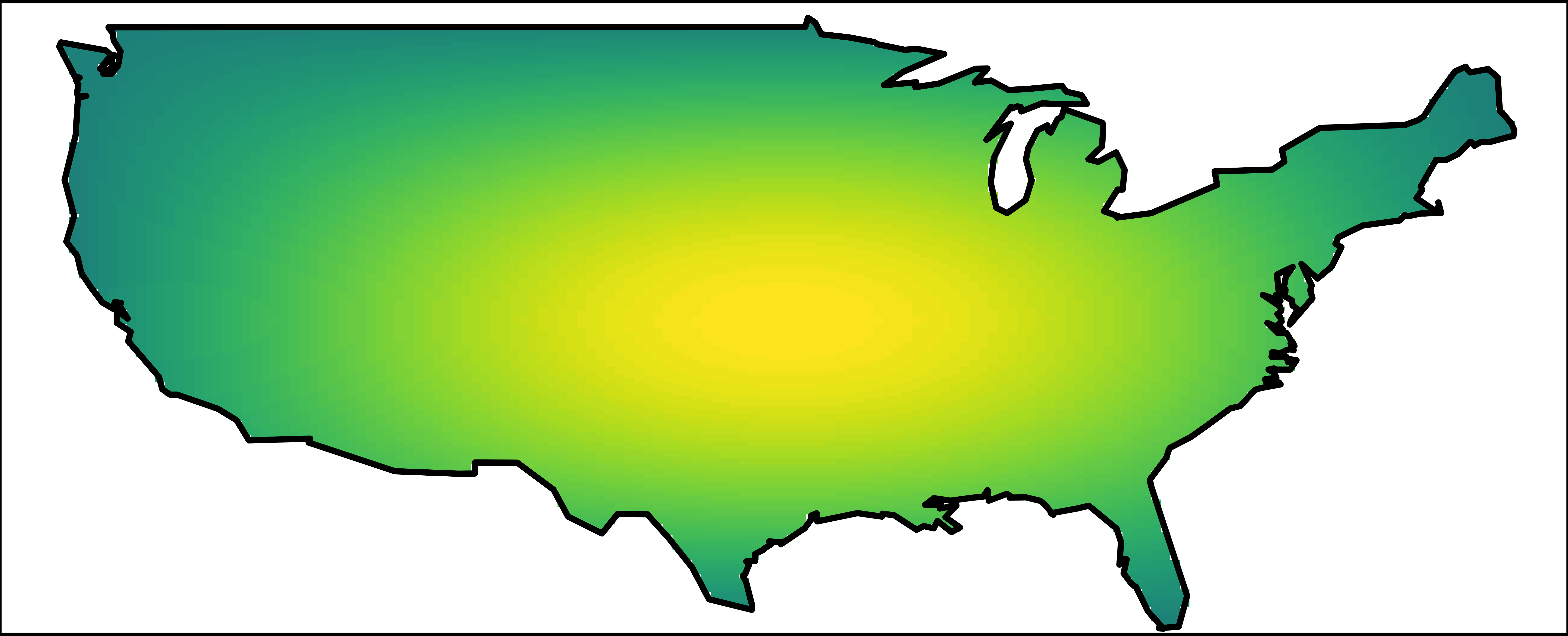}
\includegraphics[width=\linewidth]{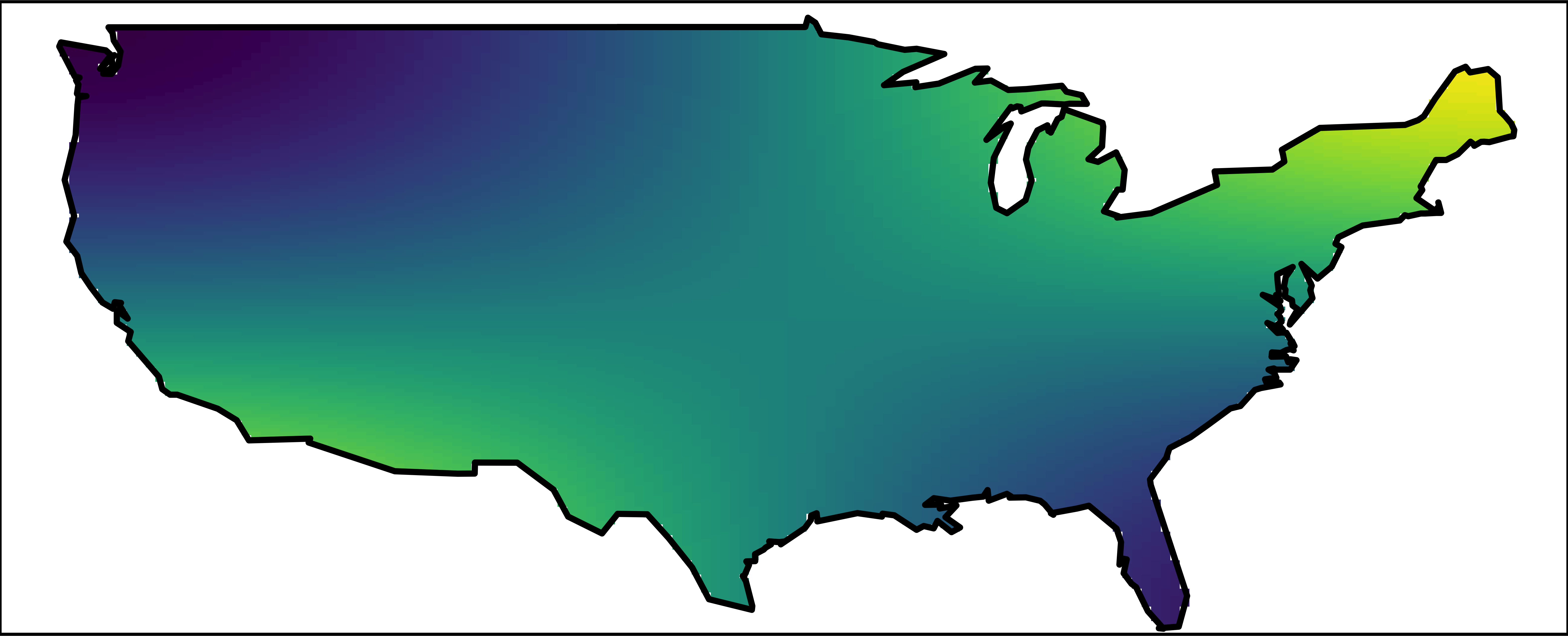}
\end{minipage}
\begin{minipage}{0.25\linewidth}
\includegraphics[width=\linewidth]{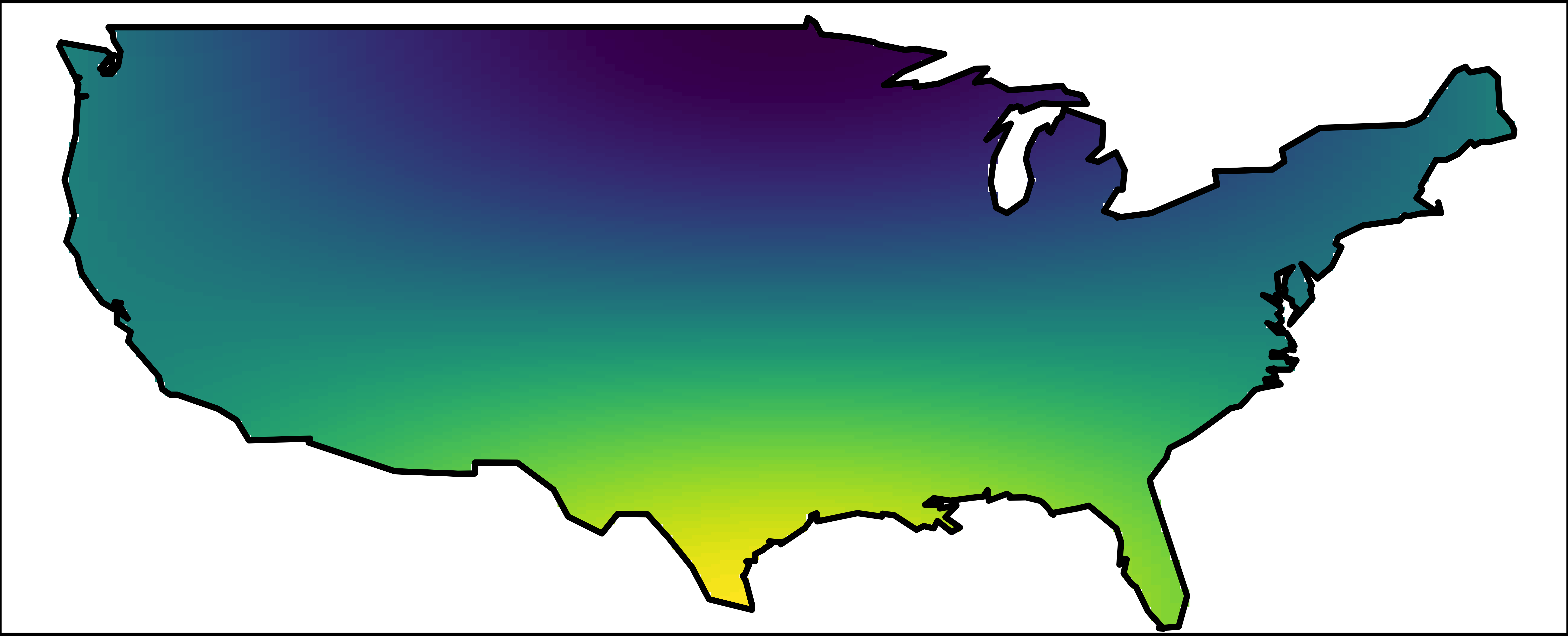}
\includegraphics[width=\linewidth]{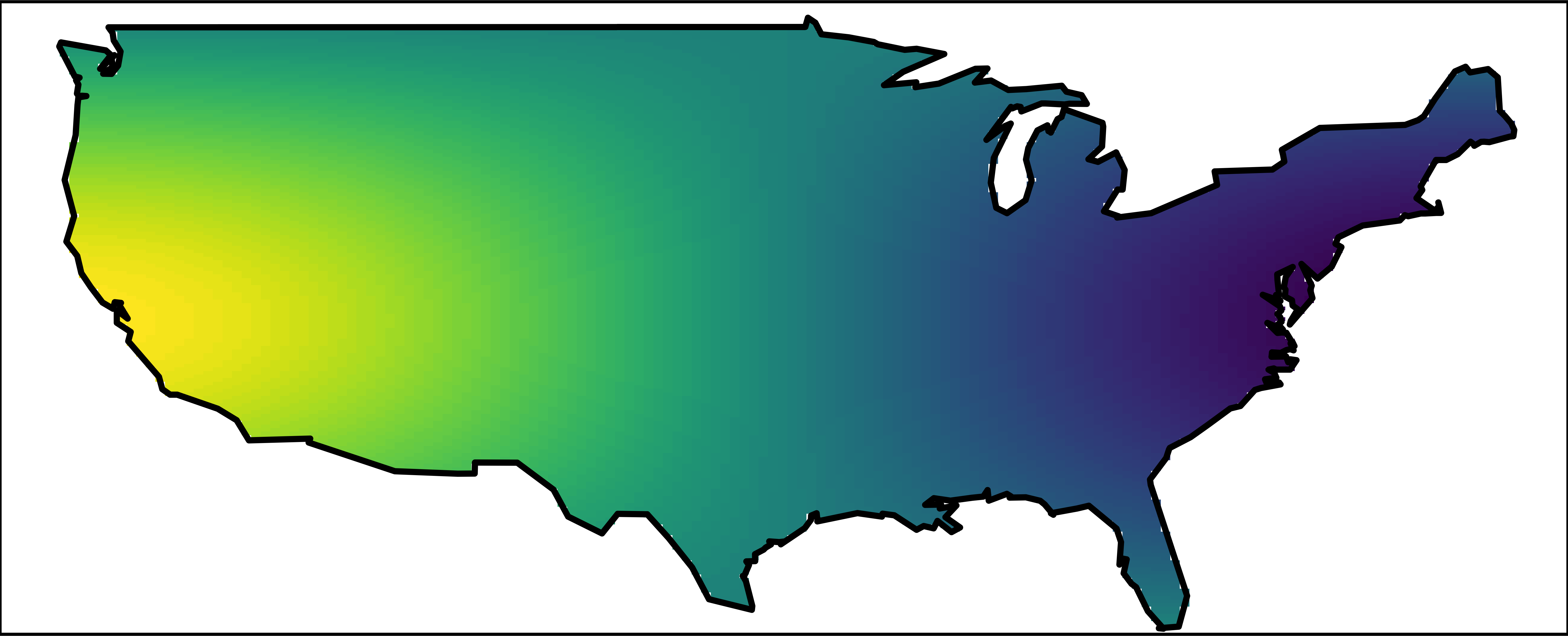}
\includegraphics[width=\linewidth]{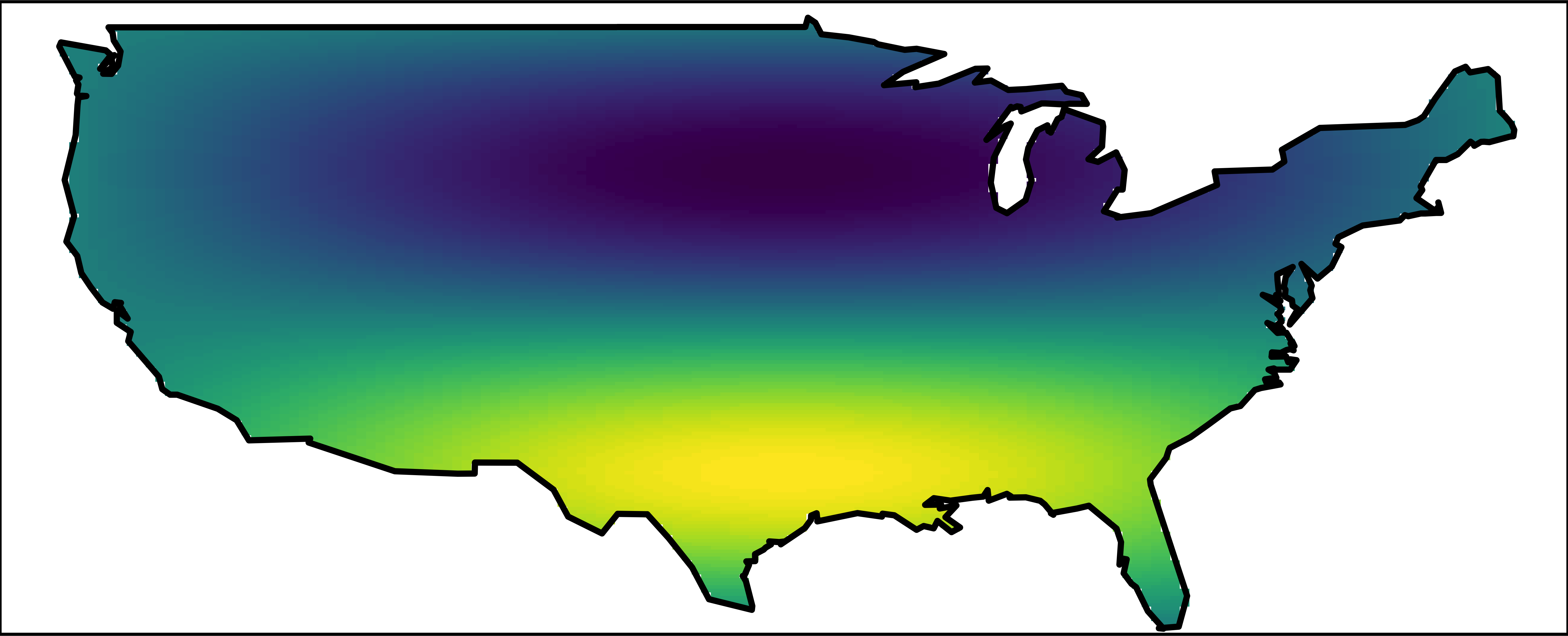}
\end{minipage}
\begin{minipage}{0.25\linewidth}
\includegraphics[width=\linewidth]{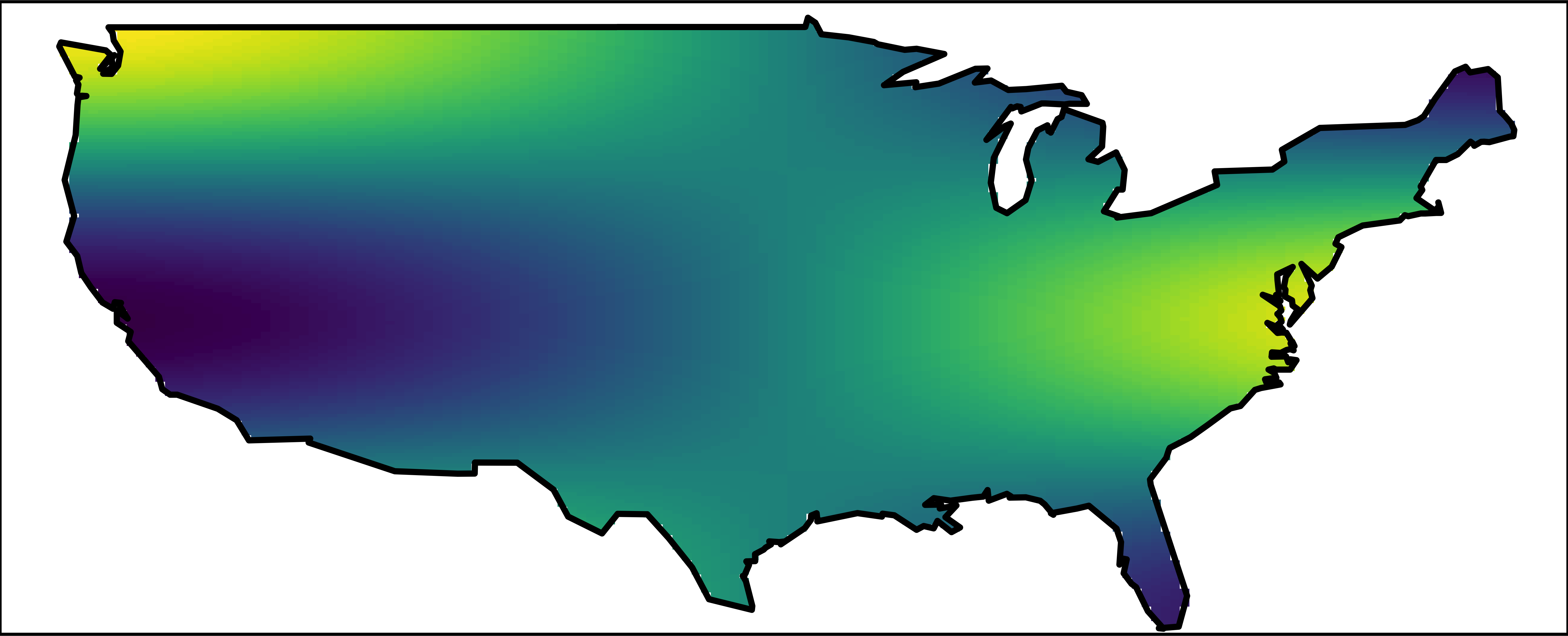}
\includegraphics[width=\linewidth]{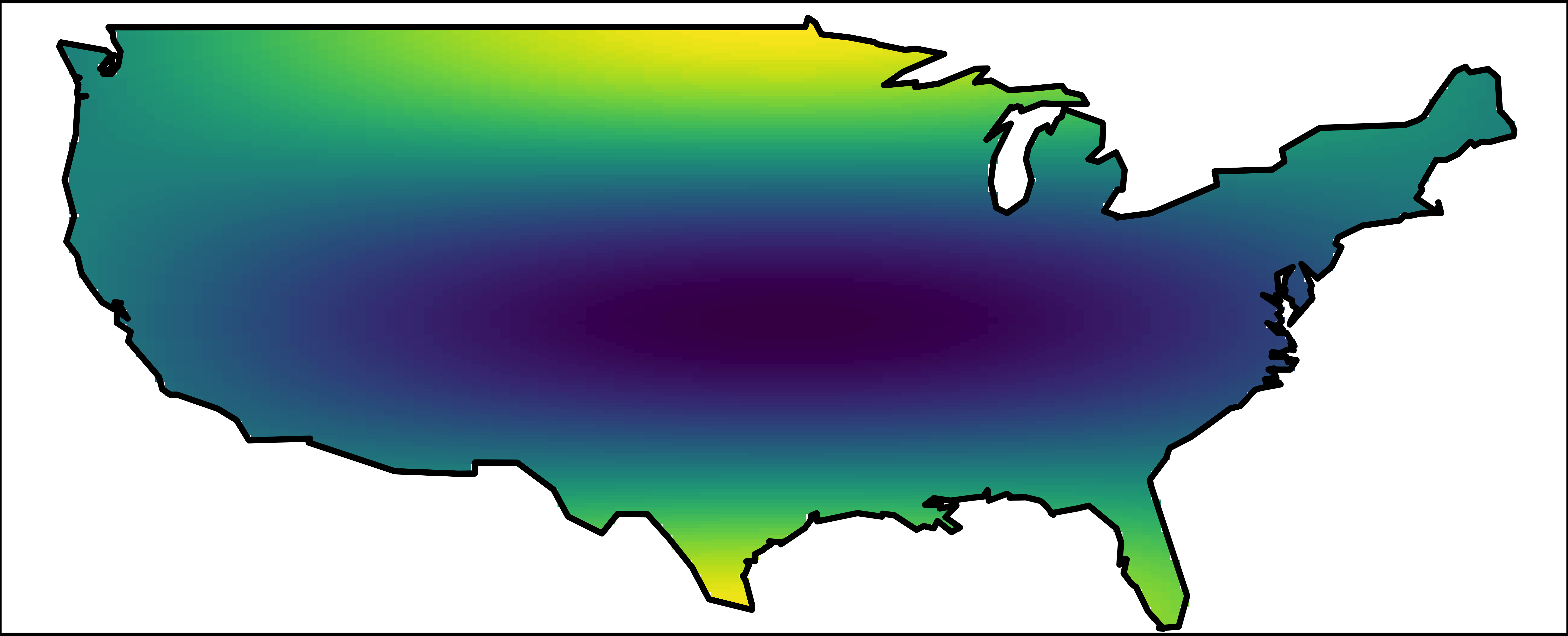}
\includegraphics[width=\linewidth]{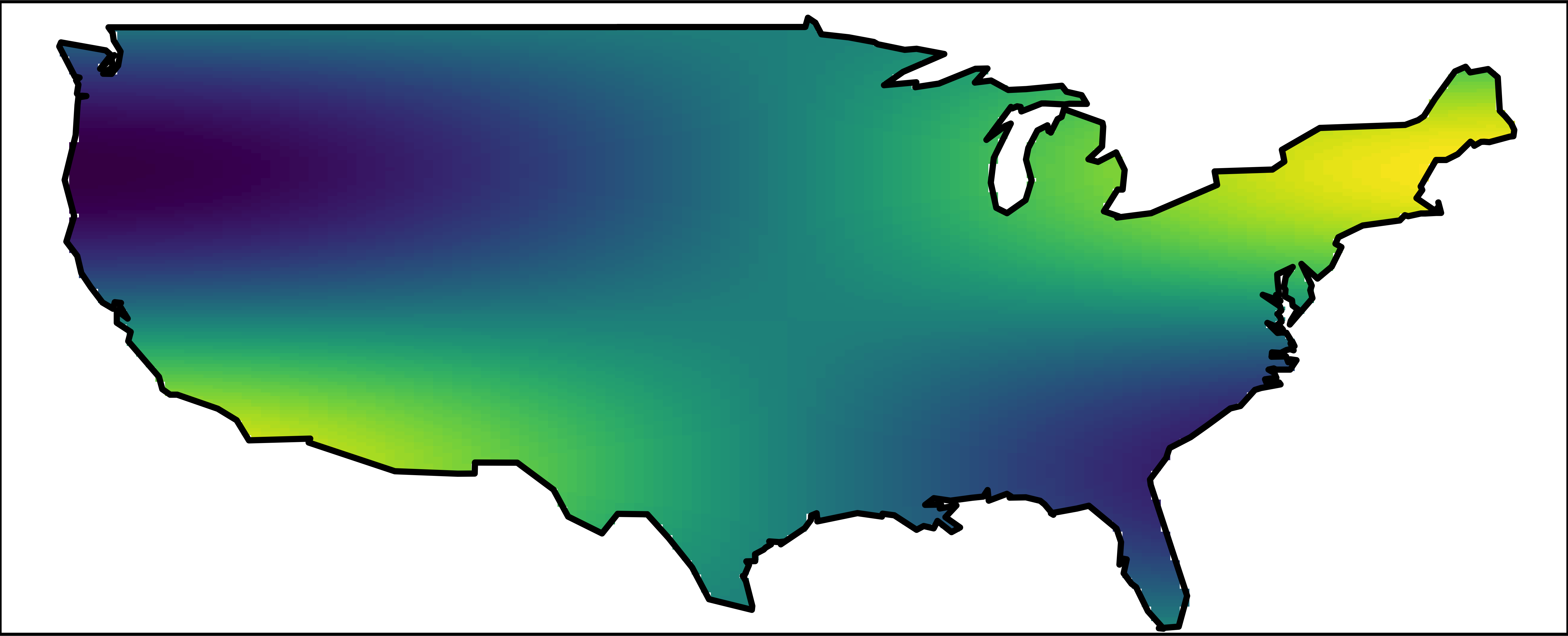}
\end{minipage}
\begin{minipage}{0.05\linewidth}
\includegraphics[width=\linewidth]{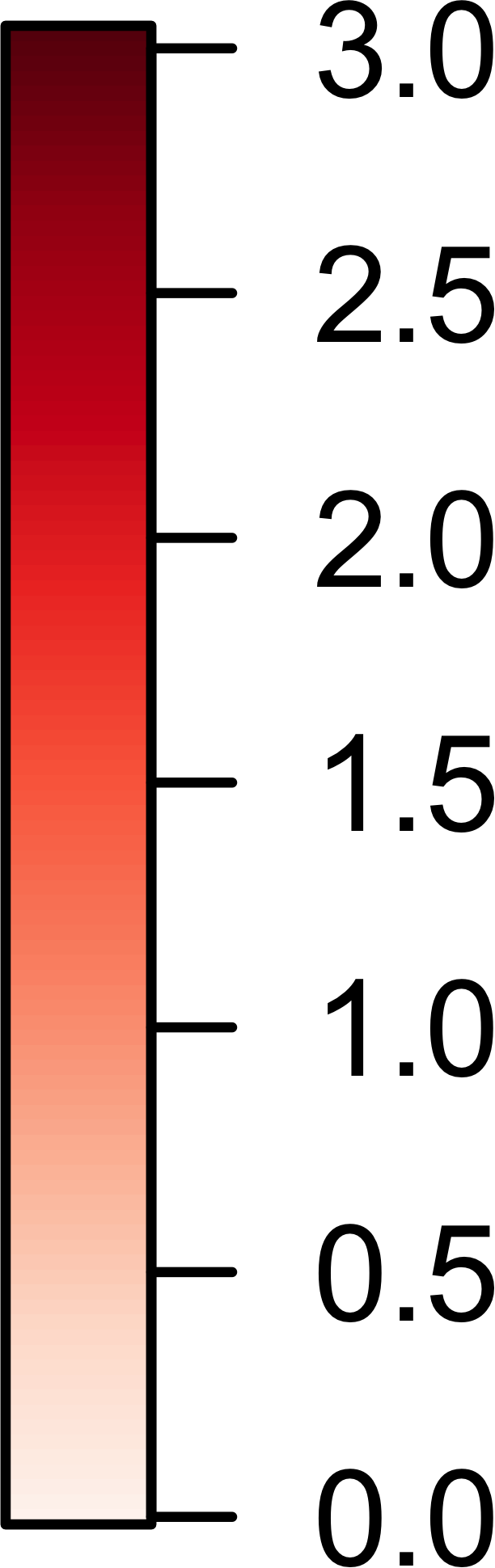}
\includegraphics[width=\linewidth]{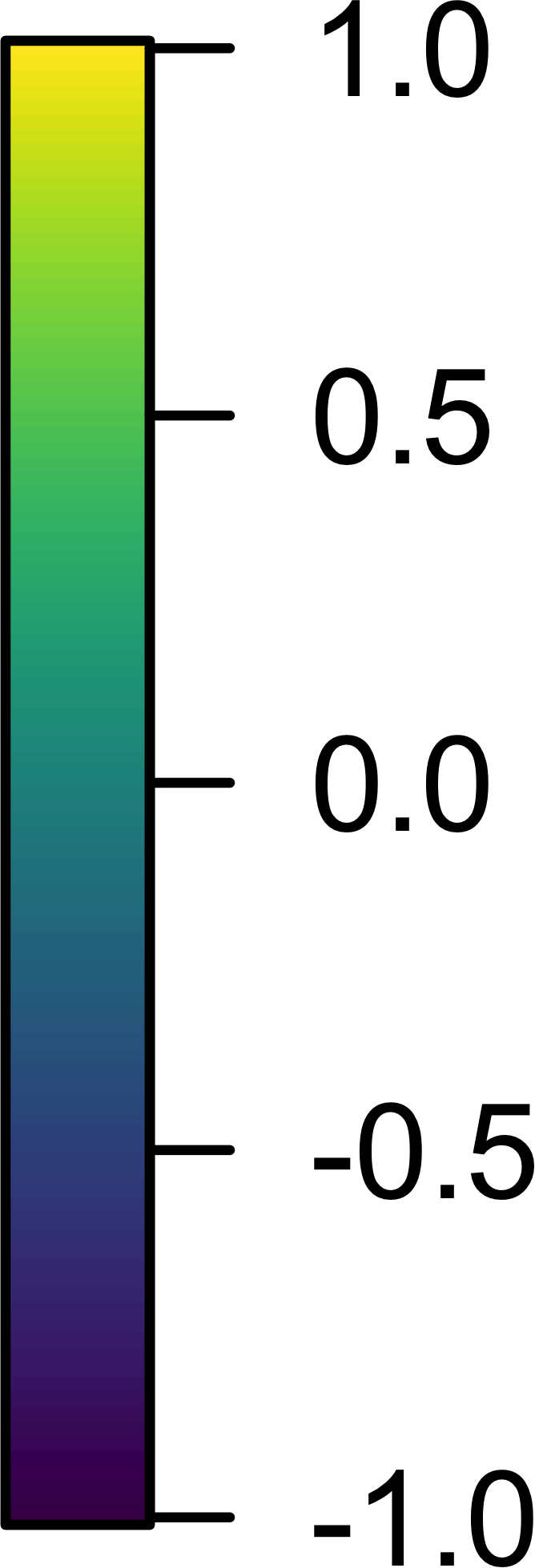}
\end{minipage}
\end{center}
\vspace{-0.2cm}
\caption{\label{fig:basis_functions}\kk{Nine basis functions 
	modeling the parameters for the non-stationary models}.}
\end{figure}

We discretize each model with respect to
the finite element mesh shown in Figure~\ref{fig:application_data},
assuming homogeneous Neumann boundary conditions.
The mesh has $5021$ nodes and was computed using R-INLA \citep{lindgren2015software}.
For the fractional models, we set $m=1$ in the rational approximation
and, for each model, the model parameters are estimated
by numerical optimization of the log-likelihood
as described in \S\ref{sec:comp}.

The log-likelihood values for the four models
can be seen in Table~\ref{tab:application}.
%
The parameter estimates for the stationary 
non-fractional ($\beta=\nu=1$) model are
$\kappa = 0.67$, $\tau = 5.44$, and $\sigma = 0.014$, 
which implies a standard deviation $\phi = 0.077$ 
and a practical range $\rho = 4.21$. 
The estimates for the fractional model 
are $\kappa = 0.20$, $\tau = 10.58$, $\sigma = 0.012$, 
and $\beta= 0.72$, corresponding 
to $\phi = 0.081$, $\rho = 9.21$, 
and a smoothness parameter $\nu = 0.44$. 
We note that the fractional model has a longer correlation range. 
\kk{This is likely to be caused 
by the non-fractional model 
underestimating the range $\rho$ 
in order to compensate for the wrong local behavior 
of the covariance function 
induced by the smoothness parameter $\nu=1$}.

\kk{Figure~\ref{fig:application_vars} shows 
the estimated marginal standard deviation $\phi(\mv{s})$  
for the two non-stationary models  
(computed using the estimates of 
the parameters for $\kappa(\mv{s})$ and $\tau(\mv{s})$) 
and $0.7$ contours of the correlation function 
for selected locations in the domain. 
The estimate of $\beta$ 
for the non-stationary fractional model is $0.723$.
Also for the non-stationary models, 
we observe a slightly longer 
practical correlation range 
$\rho(\mv{s})$ for the fractional model}. 
%

\begin{figure}[t]
\begin{center}
\begin{minipage}{0.45\linewidth}
\begin{center}
Fractional model\\
 \end{center}
\includegraphics[width=\linewidth]{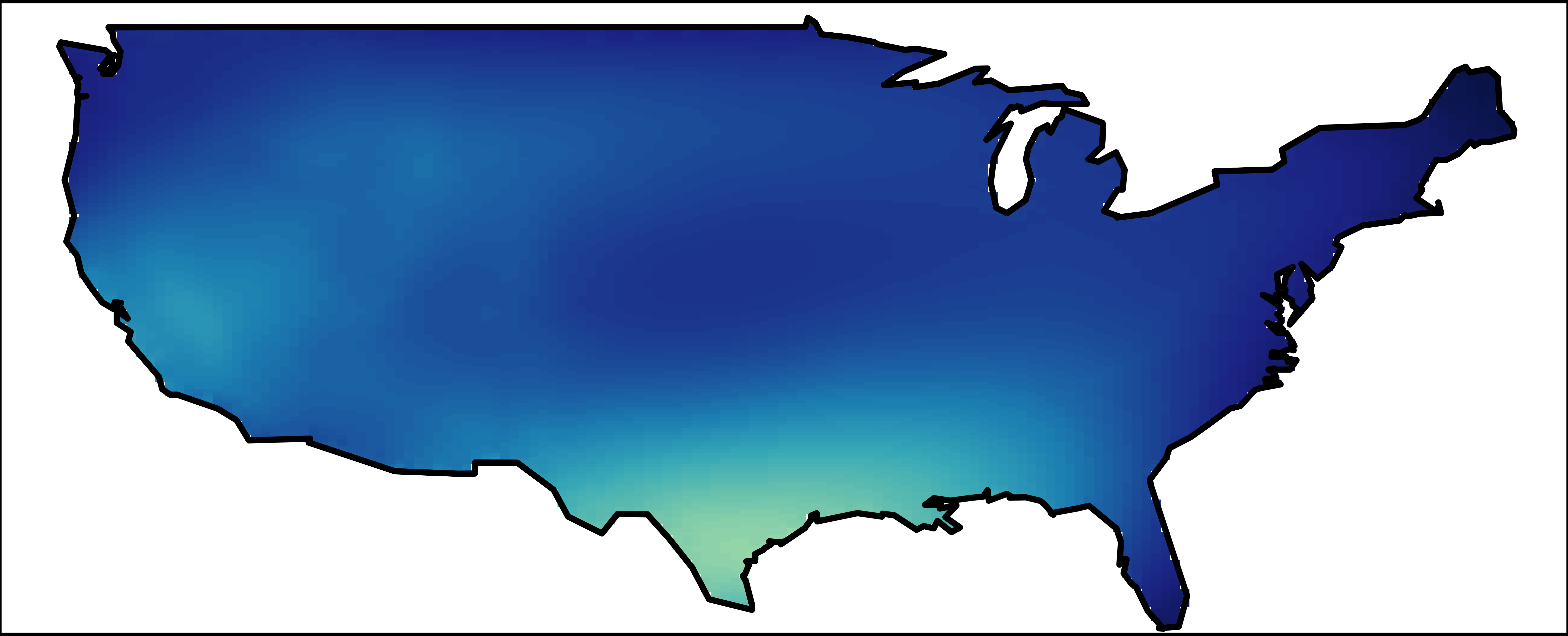}
\includegraphics[width=\linewidth]{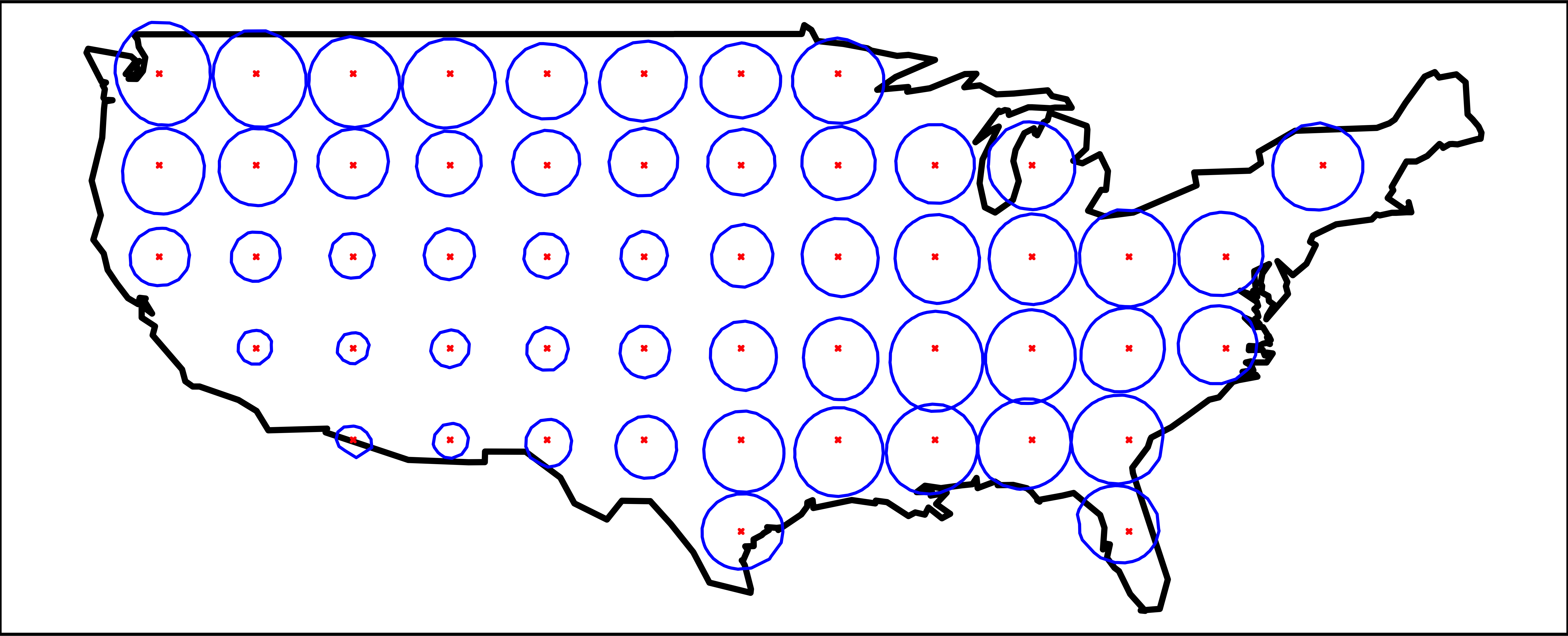}
\end{minipage}
\begin{minipage}{0.45\linewidth}
\begin{center}
$\beta=1$ model\\
\end{center}
\includegraphics[width=\linewidth]{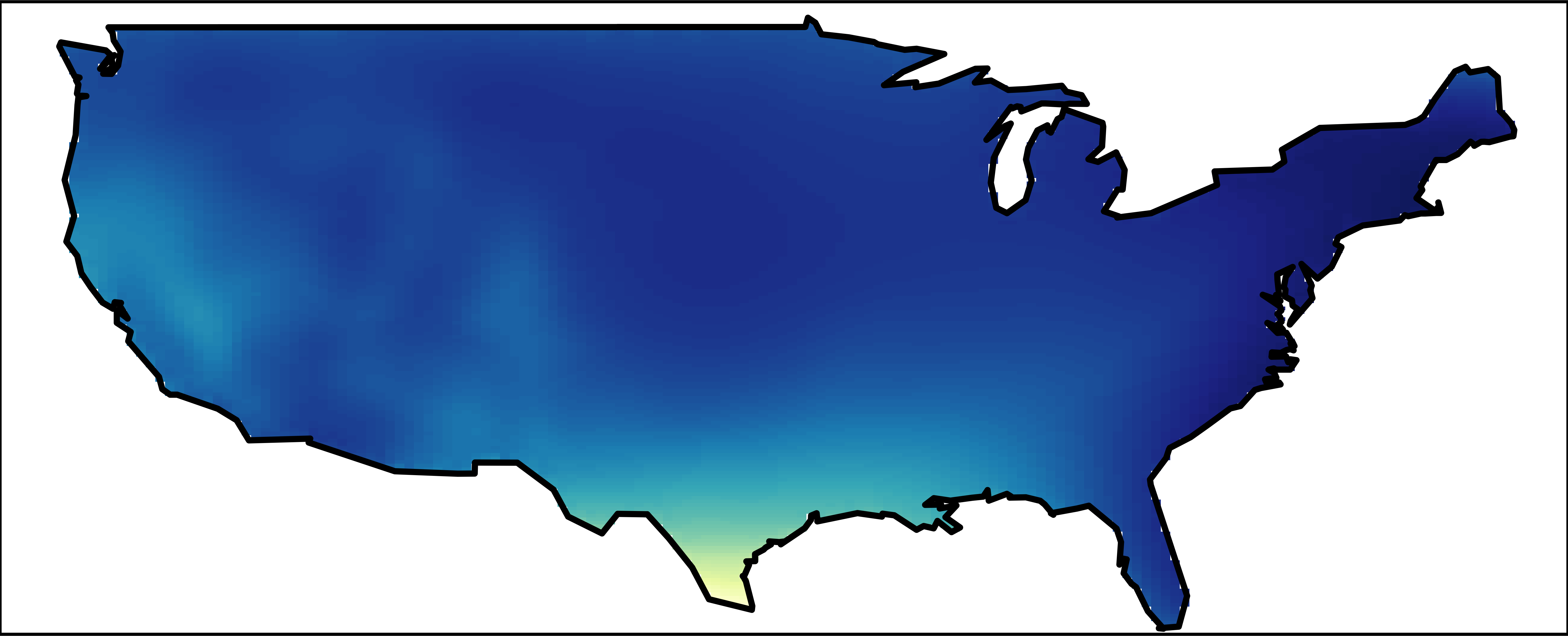}
\includegraphics[width=\linewidth]{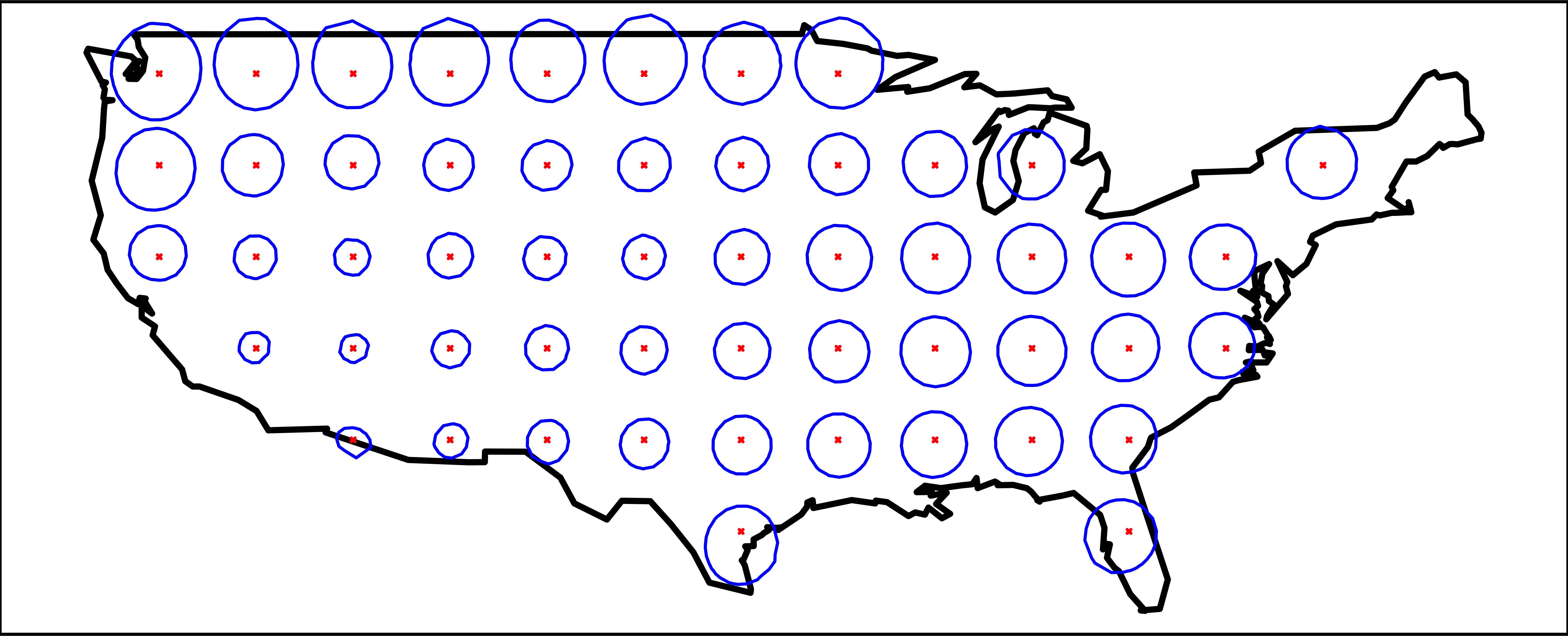}
\end{minipage}
\begin{minipage}{0.065\linewidth}
\raisebox{2cm}{\includegraphics[width=\linewidth]{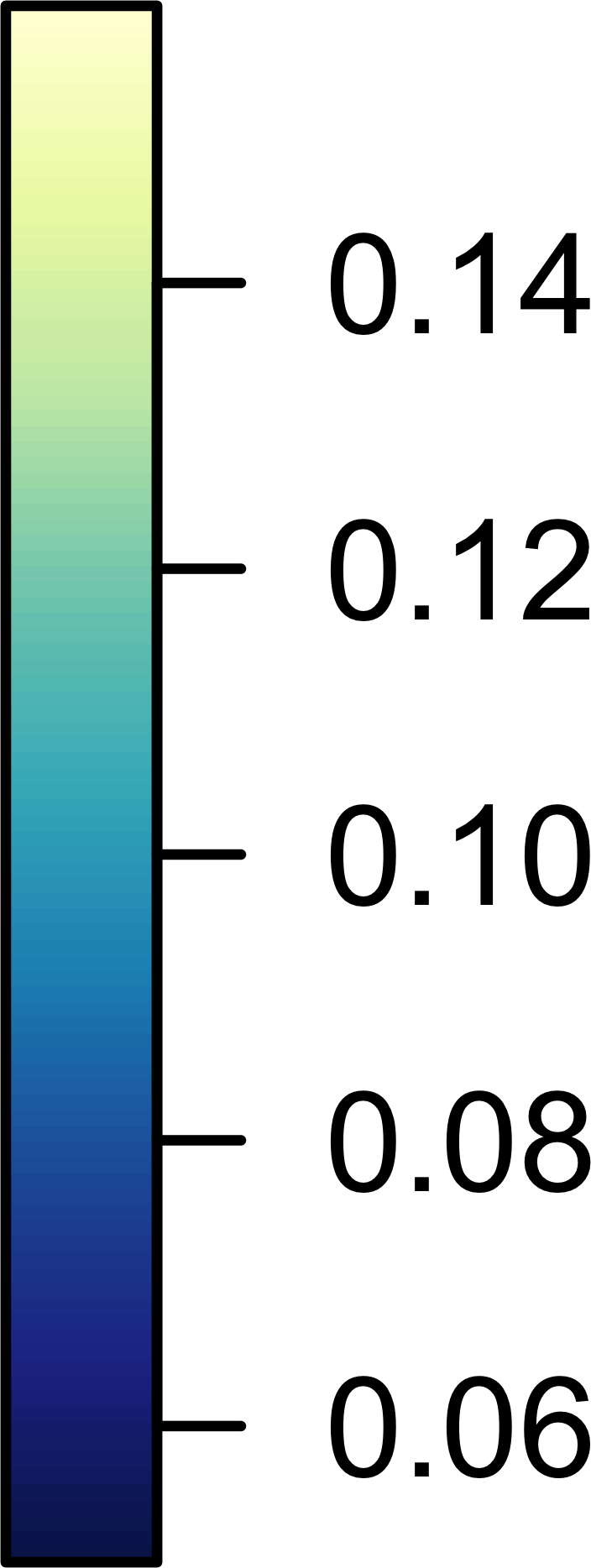}}	
\end{minipage}
\end{center}
\vspace{-0.2cm}
\caption{\label{fig:application_vars}Estimated marginal standard deviations (top row) and contours of $0.7$ correlation of the correlation function for selected locations (bottom row), for the fractional (left column) and $\beta=1$ (right column) models.}
\end{figure}

\begin{table}[t]
\centering
\begin{tabular}{lccccc}
\toprule
						& Log-likelihood 	& RMSE	& CRPS & LS & Time\\
Stationary $\beta=1$		& 219773 & 4.206 & 2.295 & 177.2 & 0.125\\
Stationary fractional		& 220255 & 4.167 & 2.274 & 178.2 & 0.412\\
Non-stationary $\beta=1$	& 225969 & 4.194 & 2.266 & 182.1 & 0.121\\
Non-stationary fractional 	& 226095 & 4.170 & 2.254 & 182.4 & 0.416\\
    \bottomrule
\end{tabular}
\vspace{0.2cm}
\caption{Model-dependent results for 
	(i) the log-likelihood, 
	(ii) the pseudo cross-validation scores 
		(RMSE, CRPS, LS, each $\times 100$) averaged over ten replicates, and 
	(iii) the computational time for one evaluation of the likelihood 
		averaged over $100$ computations.}
		\label{tab:application}
\end{table}

To investigate the predictive accuracy of the models,
a pseudo cross-validation study is performed.
We choose $10\%$ of the spatial observation locations at random,
and use the corresponding observations for each year
to predict the values at the remaining locations.
The accuracy of the four models is measured
by the root mean square error~(RMSE), 
the average continuous ranked probability score~(CRPS), 
and the average log-score (LS). 
This procedure is repeated ten times,
where in each iteration new locations are chosen at random
to base the predictions on.
\kk{The average scores for the ten iterations 
are shown in Table~\ref{tab:application}. 
Recall that low RMSE and CRPS values 
resp.\ high LS values correspond to good scores}. 

\kk{We observe} 
\db{that the predictive performance 
of the non-stationary non-fractional ($\beta=1$) model 
is similar to the stationary fractional model 
in terms of CRPS, and actually worse in terms of RMSE. 
This clearly indicates that the data should be analyzed 
by a fractional model. 
Although the non-stationary fractional model 
has a better performance in terms of CRPS and LS 
than the stationary fractional model, 
the difference is quite small 
given that the non-stationary model has $38$ parameters, 
compared to $4$ for the stationary model. 
Thus, the fractional smoothness seems to be the most important aspect for this data. 
%
The fact that the rational SPDE approach allows us 
to make these comparisons and 
\kk{to verify the smoothness parameter,  
for stationary and non-stationary models,
is one of its most important features}}.

\section{Discussion}\label{sec:discussion}
We have introduced the rational SPDE approach
providing a new type of computationally efficient approximations
for a class of Gaussian random fields.
These are based on an extension of the SPDE approach by \cite{lindgren11}
to models with general second-order differential operators of arbitrary order $\beta > d/4$
For these approximations, explicit rates of strong convergence have been derived
and we have shown how to calibrate the degree of the rational approximation
with the mesh size of the FEM to achieve these rates. \db{The results can also be combined with the results in \citep{bolin2018weak} to obtain explicit rates of weak convergence (convergence of functionals of the random field).} 

Our approach can, e.g., be used to approximate
stationary Mat\'ern fields with general smoothness, and it is also directly 
applicable
to more complicated non-stationary models, 
where the covariance function may be unknown.
A general fractional order of the differential operator
opens up for new applications of the SPDE approach,
such as to Gaussian fields with exponential covariances on $\mathbb{R}^2$.
For the Mat\'ern model and its extensions,
it furthermore facilitates likelihood-based (or Bayesian)
inference of all model parameters.
The specific structure of the approximation
then in turn enables a combination with INLA or MCMC
in situations where the Gaussian model is a part
of a more complicated non-Gaussian hierarchical model.

We have illustrated the rational SPDE approach 
for stationary and non-stationary Mat\'ern models.
A topic for future research is to apply the approach 
to other random field models in statistics 
which are difficult to approximate by GMRFs,
such as to models with long-range dependence \citep{lilly2017fractional}
based on the fractional Brownian motion.
%
Another topic for future research is to modify the
fractional SPDE approach by replacing the FEM basis
by a multiresolution basis and to compare this approach
to other multiresolution approaches such as \citep{katzfuss2017multi}.
Finally, it is also of interest to extend the method to non-Gaussian versions
of the SPDE-based Mat\'ern models \citep{wallin15},
since the Markov approximation considered by \cite{wallin15} is only computable
under the restrictive requirement $\beta \in \mathbb{N}$.

\begin{appendix}
\section{Iterated finite element method}\label{app:iter-fem}

The rational approximation $u_{h,m}^R$
of the solution $u$ to~\eqref{e:Lbeta} introduced in \S\ref{subsec:rat-approx}
is defined in terms of the discrete operators $\lop = p_\ell(L_h)$
and $\rop = p_r(L_h)$ via~\eqref{e:uhr}.
Since the differential operator $L$
\kk{in~\eqref{e:L-div}} is of
second order,
their continuous counterparts
$\lopex = p_\ell(L)$ and $\ropex=p_r(L)$ in~\eqref{e:ur} are
differential operators
of order $2(m+m_\beta)$ and $2m$, respectively.
Using a standard Galerkin approach for solving~\eqref{e:ur} would
therefore require finite element basis functions~$\{\varphi_j\}$
in the Sobolev space $H^{m+m_\beta}(\cD)$,
which are difficult to construct in more than one space dimension.
This can be avoided by using
a modified version of
the iterated Hilbert space approximation method
by \cite{lindgren11},
and in this section we give the details of this procedure.

Recall from \S\ref{subsec:discrete} that
$V_h \subset V$ is a finite element space with continuous
piecewise \kk{linear} basis functions $\{\varphi_j\}_{j=1}^{n_h}$
defined with respect to
a regular triangulation~$\cT_h$ of the domain $\overline{\cD}$
with mesh \kk{width} $h := \max_{T\in\cT_h} \operatorname{diam}(T)$.

For computing the finite element approximation, we start by factorizing
the polynomials~$q_1$ and $q_2$ in the rational approximation
$\hat{r}$ of $\hat{f}(x) = x^{\beta-m_\beta}$
in terms of their roots,
\begin{align*}
    q_1(x) = \sum_{i=1}^m c_i x^i = c_m \prod_{i=1}^{m} (x - r_{1i} )
    \quad
    \text{and}
    \quad
    q_2(x) = \sum_{j=1}^{m+1} b_j x^j = b_{m+1} \prod_{j=1}^{m+1} (x - r_{2j} ).
\end{align*}
We use these expressions to reformulate~\eqref{e:xbeta} as
\begin{align*}
    x^{-\beta} = f(x^{-1}) \approx \hat{r}(x^{-1}) x^{- m_\beta}
    = \frac{c_m \prod_{i=1}^{m} (1 - r_{1i}x )}{ b_{m+1} x^{m_\beta-1} \prod_{j=1}^{m+1} (1 - r_{2i}x ) },
\end{align*}
where, again, we have expanded the fraction with $x^m$.
This representation shows that we can equivalently
define the rational SPDE approximation $u_{h,m}^R$
as the solution to~\eqref{e:uhr}
with $\lop,\rop$ redefined as
$\lop = b_{m+1} L_h^{m_\beta-1} \prod_{j=1}^{m+1} (\kk{\mathrm{Id}_{h}} - r_{2j} L_h)$ 
and 
$\rop = c_m \prod_{i=1}^{m} ( \kk{\mathrm{Id}_{h}} - r_{1i} L_h )$, 
\kk{where $\mathrm{Id}_{h}$ denotes the identity on $V_h$}. 

We use the formulation of \eqref{e:uhr}
as a system outlined in \eqref{e:nested-discrete}:
First we solve $\lop x_{h,m} = \white_h$
and we then compute 
$u_{h,m}^R = \rop x_{h,m}$.
To this end,
we define the functions $x_k \in L_2(\Omega;V_h)$ 
for $k\in\{1,\ldots,m+m_\beta\}$ 
iteratively by
\begin{align*}
    b_{m+1} ( \kk{\mathrm{Id}_{h}} - r_{21} L_h ) x_1 &= \white_h, \\
    ( \kk{\mathrm{Id}_{h}} - r_{2k} L_h ) x_k &= x_{k-1}, \qquad k = 2,\ldots,m+1, \\
    L_h x_k &= x_{k-1}, \qquad k = m+2,\ldots,m+m_\beta, \quad \mbox{if $m_{\beta}\geq2$,}
\end{align*}
\kk{noting} that $x_{m+m_\beta} = x_{h,m}$.

By 
\kk{recalling the bilinear form $a_L$ from~\eqref{e:a-L} 
and}
expanding $x_k = \sum_{j=1}^{n_h} x_{kj} \varphi_j$
with respect to the finite element basis,
we find that the stochastic weights
\kk{$\mv{x}_{k} = (x_{k1}, \ldots, x_{k{n_h}})^{\trsp}$
satisfy
{\allowdisplaybreaks
\begin{gather*}
    \sum_{j=1}^{n_h} x_{1j} \, b_{m+1} 
    \left(\scalar{\varphi_j, \varphi_i}{L_2(\cD)} - r_{21} \, a_L(\varphi_j, \varphi_i) \right) 
    	= \scalar{\white_h, \varphi_i}{L_2(\cD)},  \\
    \sum_{j=1}^{n_h} x_{kj} 
    \left( \scalar{\varphi_j, \varphi_i}{L_2(\cD)} - r_{2k} \, a_L(\varphi_j, \varphi_i) \right)
        = \sum_{j=1}^{n_h} x_{k-1,j} \, \scalar{\varphi_j, \varphi_i}{L_2(\cD)},
     \,\, 2\leq k \leq m+1,
    \\
    \sum_{j=1}^{n_h} x_{kj} \,  
    a_L(\varphi_j, \varphi_i) 
        = \sum_{j=1}^{n_h} x_{k-1,j} \, \scalar{ \varphi_j, \varphi_i}{L_2(\cD)},
    \quad 
k = m+2,\ldots,m+m_\beta,
\end{gather*}}
where} each of these equations hold for $i = 1,\ldots, n_h$.
Recall from \kk{\S\ref{subsec:discrete}} that
$\white_h$ is white noise in $V_h$.
This entails the distribution 
\kk{$\bigl(\scalar{\white_h, \varphi_i}{L_2(\cD)} \bigr)_{i=1}^{n_h} \sim \pN(\mv{0},\mv{C})$},
where $\mv{C}$ is the mass matrix with elements
$C_{ij} = \scalar{\varphi_j, \varphi_i}{L_2(\cD)}$
and, therefore, 
\kk{$\mv{x}_k \sim 
	\pN\bigl(\mv{0}, 
		\mv{P}_{\ell,k}^{-1} \mv{C} \mv{P}_{\ell,k}^{-\trsp} 
	\bigr)$} 
for every $k \in\{ 1,...,m+m_\beta\}$.
Here, the matrix \kk{$\mv{P}_{\ell,k}$ is defined by
\begin{align*}
    \mv{P}_{\ell,k} = 
    	\begin{cases}
            b_{m+1} \mv{C} \, \mv{L}_{k}, 
            & k=1,\ldots, m+1, \\
            b_{m+1} \mv{C} \left( \mv{C}^{-1} \mv{L} \right)^{k-m-1} 
            	\mv{L}_{m+1}, 
            	& k = m+2, \ldots , m+m_\beta,
        \end{cases}
\end{align*}
where $\mv{L}_{k} 
:= \prod_{j=1}^{k} \left( \mv{I} - r_{2j}\mv{C}^{-1}\mv{L} \right)$, 
with identity matrix $\mv{I} \in \bbR^{n_h \times n_h}$}, 
and the entries of $\mv{L}$ are given \kk{by 
\begin{align*}
	L_{ij} 
	:= 
	a_L(\varphi_j, \varphi_i) 
	= 
	 \scalar{\mv{H} \nabla\varphi_j, \nabla\varphi_i}{L_2(\cD)}
	+ 
	\left( \kappa^2 \varphi_j, \varphi_i \right)_{L_2(\cD)}, 
	\qquad 
	i,j = 1, \ldots, n_h,  	
\end{align*}
cf.~\eqref{e:L-div}--\eqref{e:a-L}}. 
In particular, the weights $\mv{x}$ of $x_{h,m}$ have distribution
\begin{align}\label{e:xdistribution}
    \mv{x} \sim 
    \pN\left( \mv{0}, \mv{P}_{\ell}^{-1} \mv{C} \mv{P}_{\ell}^{-\trsp} \right),
    \qquad
    \text{where}
    \qquad
    \mv{P}_{\ell} := \mv{P}_{\ell,m+m_\beta}.
\end{align}
Note also that for the Mat\'ern \kk{case, i.e., $L = \kappa^{2} - \Delta$},
we have $\mv{L} = \kappa^{2}\mv{C} + \mv{G}$,
where $\mv{G}$ is the stiffness matrix with elements
$G_{ij} = \scalar{\nabla\varphi_j, \nabla\varphi_i}{L_2(\cD)}$.

To calculate the final approximation
$u_{h,m}^R = \rop x_{h,m}$,  
we apply a similar iterative procedure.
Let $u_1,\ldots,u_m$ be defined by
\begin{align*}
    u_1 &= c_m (\kk{\mathrm{Id}_h} - r_{11}L_h) x_{h,m}, \\
    u_k &= (\kk{\mathrm{Id}_h} - r_{1k}L_h) u_{k-1}, \qquad\quad k = 2,\ldots,m.
\end{align*}
Then $u_{h,m}^R = c_m \bigl(\prod_{i=1}^m (\mathrm{Id} - r_{1i} L_h) \bigr)x_{h,m} = u_m$
and the weights $\mv{u}_{k}$ of $u_k$
can be obtained from the weights of $x_{h,m}$ via
$\mv{u}_{k} = \mv{P}_{r,k} \, \mv{x}$,
where 
$$
\mv{P}_{r,k} 
:= c_m \prod_{i=1}^{k} \left( \mv{I} - r_{1i} \mv{C}^{-1} \mv{L} \right).
$$
By~\eqref{e:xdistribution}, the distribution
of the weights $\mv{u}$ of the final rational approximation $u_{h,m}^R$
is thus given by
\begin{align*}
    \mv{u} \sim 
    \pN\left( \mv{0}, 
    	\mv{P}_{r} \mv{P}_{\ell}^{-1} \mv{C} 
    		\mv{P}_{\ell}^{-\trsp} \mv{P}_{r}^{\trsp} \right),
    \qquad
    \text{where}
    \qquad
    \mv{P}_{r} := \mv{P}_{r,m}.
\end{align*}

To obtain sparse matrices $\mv{P}_{\ell}$ 
and $\mv{P}_{r}$,
we approximate the mass matrix $\mv{C}$ by a diagonal matrix
$\widetilde{\mv{C}}$ with diagonal elements
$\widetilde{C}_{ii} = \sum_{j=1}^{n_h} C_{ij}$.
The effect of this ``mass lumping''
was motivated theoretically by \cite{lindgren11},
and was empirically shown to be small by \cite{bolin13comparison}.

\section{Convergence analysis}\label{app:convergence}
In this section we give the details of the convergence result
stated in Theorem~\ref{thm:strong}.
As mentioned in \S\ref{subsec:error},
we choose $\hat{r}=\hat{r}_h$ as the $L_\infty$-best
rational approximation of $\hat{f}(x) = x^{\beta - m_\beta}$
on the interval $J_h$ for each $h$. We furthermore assume that
the operator $L$ \kk{in~\eqref{e:L-div}} 
is normalized such that $\lambda_1 \geq 1$
and, thus, $J_h \subset J \subset [0,1]$.

Recall that Proposition~\ref{prop:uh}
provides a bound for $\norm{u - u_h}{L_2(\Omega; L_2(\cD))}$.
Therefore, it remains is to estimate the strong error
between $u_{h,m}^R$ and $u_h$
induced by the rational approximation of $f(x) = x^\beta$.
To this end, recall the construction
of the rational approximation $u_{h,m}^R$ 
\kk{from} \S\ref{subsec:rat-approx}:
We first decomposed $f$
as $f(x) = \hat{f}(x) x^{m_\beta}$,
where $\hat{f}(x) = x^{\beta-m_\beta}$, and
then used  a rational approximation
$\hat{r} = \frac{q_1}{q_2}$ of $\hat{f}$
on the interval $J_h= \bigl[ \lambda_{n_h,h}^{-1}, \lambda_{1,h}^{-1} \bigr]$
with $q_1 \in \cP^m(J_h)$ and $q_2 \in \cP^{m+1}(J_h)$
to define the approximation $r(x) := \hat{r}(x) x^{m_\beta}$ of $f$.
Here, $\cP^m(J_h)$ denotes the set of polynomials
$q\from J_h \to \bbR$ of degree $\deg(q) = m$.
In the following, we assume that $\hat{r} = \hat{r}_h$ is the best rational approximation
of $\hat{f}$ of this form, i.e.,
\begin{align*}
    \norm{\hat{f} - \hat{r}_h}{C(J_h)}
    =
    \inf\left\{
    \norm{\hat{f} - \hat{\rho}}{C(J_h)}
    :
    \hat{\rho} = \tfrac{q_1}{q_2}, \
    q_1\in\cP^m(J_h), \
    q_2\in \cP^{m+1}(J_h) \right\},
\end{align*}
where $\norm{g}{C(J)} := \sup_{x\in J} |g(x)|$.

For the analysis, we treat the two cases
$\beta\in(0,1)$ and $\beta \geq 1$ separately.
If $\beta \geq 1$, then $\hat{\beta} := \beta - m_\beta \in [0,1)$.
Thus, if $\hat{r}_{*}$ denotes the best rational approximation
of~$\hat{f}$ on the interval $[0,1]$,
we find
\citep[][Theorem~1]{stahl2003rational}
\begin{align*}
    \norm{ \hat{f} - \hat{r}_h }{C(J_h)}
    \leq \sup_{x\in[0,1]} | \hat{f}(x) - \hat{r}_{*}(x) |
        \leq \hat{C} e^{-2\pi \sqrt{\hat{\beta}m}},
\end{align*}
where the constant $\hat{C}>0$ is continuous in $\hat{\beta}$ and independent of $h$ and the degree~$m$.
Since $x^{m_\beta} \leq 1$ for all $x\in J_h$, we obtain for $r_h(x) := \hat{r}_h(x) x^{m_\beta}$
the same bound,
\begin{align}\label{e:r-beta>1}
    \norm{ f - r_h }{C(J_h)}
    \leq \sup_{x\in J_h} | \hat{f}(x) - \hat{r}_h(x) |
        \leq \hat{C} e^{-2\pi \sqrt{\hat{\beta}m}}.
\end{align}

If $\beta\in(0,1)$, then $\hat{\beta} \in (-1,0)$ and
we let $\widetilde{r}$
be the best approximation of $\widetilde{f}(x) := x^{|\hat{\beta}|}$
on $[0,1]$.
A rational approximation of $\widetilde{f}$
on the different interval
$\widetilde{J}_h := [\lambda_{1,h},\lambda_{n_h,h}]$
is then given by
$\widetilde{R}_h(\widetilde{x}) := \lambda_{n_h,h}^{|\hat{\beta}|} \widetilde{r}(\lambda_{n_h,h}^{-1} \widetilde{x})$
with error 
\begin{align*}
    \sup_{\widetilde{x}\in \widetilde{J}_h} | \widetilde{f}(\widetilde{x}) - \widetilde{R}_h(\widetilde{x}) |
        \leq \lambda_{n_h,h}^{|\hat{\beta}|} \sup_{x\in[0,1]} | \widetilde{f}(x) - \widetilde{r}(x) |
        \leq \widetilde{C} \lambda_{n_h,h}^{|\hat{\beta}|} e^{-2\pi \sqrt{|\hat{\beta}|m}},
\end{align*}
where the constant $\widetilde{C} > 0$ depends only on $|\hat{\beta}|$.
On $J_h = \bigl[ \lambda_{n_h,h}^{-1}, \lambda_{1,h}^{-1} \bigr]$
the function $\widetilde{R}_h(x^{-1})$ is
an approximation of $\hat{f}(x) = x^{\hat{\beta}} = \widetilde{f}(x^{-1})$ and
\begin{align*}
    \norm{\hat{f} - \hat{r}_h}{C(J_h)}
    \leq \sup_{x\in J_h} | \hat{f}(x) - \widetilde{R}_h(x^{-1}) | 
    \leq \sup_{\widetilde{x} \in \widetilde{J}_h} | \widetilde{f}(\widetilde{x}) - \widetilde{R}_h(\widetilde{x}) |
    \leq \widetilde{C} \lambda_{n_h,h}^{|\hat{\beta}|} e^{-2\pi \sqrt{|\hat{\beta}|m}}.
\end{align*}
Finally, we use again the estimate $x^{m_\beta} \leq 1$ on $J_h$
to derive 
\begin{align}\label{e:r-beta<1}
    \norm{ f - r_h }{C(J_h)}
    \leq \norm{\hat{f} - \hat{r}_h}{C(J_h)}
        \leq \widetilde{C} \lambda_{n_h,h}^{|\hat{\beta}|} e^{-2\pi \sqrt{|\hat{\beta}|m}}.
\end{align}

Proposition~\ref{prop:uh} and the estimates \eqref{e:r-beta>1}--\eqref{e:r-beta<1}
yield Theorem~\ref{thm:strong}, which is proven below.

\begin{proof}[Proof of Theorem~\ref{thm:strong}]
By Proposition~\ref{prop:uh}, it suffices
to bound $\bbE \norm{u_h - u_{h,m}^R}{L_2(\cD)}^2$.
To this end, let $\white_h = \sum_{j=1}^{n_h} \xi_j e_{j,h}$
be a Karhunen--Lo\`{e}ve expansion of $\white_h$, 
where $\{e_{j,h}\}_{j=1}^{n_h}$ are $L_2(\cD)$-orthonormal eigenvectors of $L_h$ 
corresponding to the eigenvalues
$\{\lambda_{j,h}\}_{j=1}^{n_h}$ and $\xi_j \sim \pN(0,1)$ i.i.d..

By construction and owing to boundedness and invertibility of $L_h$,
we have for~$u_{h,m}^R$ in~\eqref{e:uhr}
that $u_{h,m}^R = \lop^{-1} \rop \white_h = r_h(L_h^{-1}) \white_h$
and we estimate
\begin{align*}
    \bbE \norm{u_h - u_{h,m}^R}{L_2(\cD)}^2
        &=
        \bbE \sum_{j=1}^{n_h} \xi_j^2 \left( \lambda_{j,h}^{-\beta} - r_h(\lambda_{j,h}^{-1}) \right)^2
        \leq 
        n_h \max_{1\leq j \leq n_h} \bigl| \lambda_{j,h}^{-\beta} - r_h(\lambda_{j,h}^{-1}) \bigr|^2.
\end{align*}
By~\eqref{e:r-beta>1} and~\eqref{e:r-beta<1}, we can bound the last term \kk{by
\begin{align*}
\max_{1\leq j \leq n_h} \bigl| \lambda_{j,h}^{-\beta} - r_h(\lambda_{j,h}^{-1}) \bigr|^2
    &\leq  
    \left( \sup_{x\in J_h}
     |f(x) - r_h(x)| \right)^2 
    \lesssim  
    \lambda_{n_h,h}^{2\max\{(1 - \beta),0\}} e^{-4\pi \sqrt{|\beta-m_\beta| m}} .
\end{align*}
By \cite[Theorem~6.1]{strang2008} 
we have $\lambda_{n_h,h} \lesssim \lambda_{n_h} \lesssim n_h^{2/d}$, 
for sufficiently small $h\in(0,1)$, 
where the last bound follows 
from the Weyl asymptotic~\eqref{e:lambdaj}. 
Finally, $n_h\lesssim h^{-d}$ 
by quasi-uniformity of the triangulation~$\cT_h$. 
Thus, we conclude 
\begin{align*}
    \bbE \norm{u_h - u_{h,m}^R}{L_2(\cD)}^2
        &\lesssim h^{-4\max\{(1 - \beta),\, 0\} - d} e^{-4\pi \sqrt{|\beta-m_\beta| m}},
\end{align*}
which combined with Proposition~\ref{prop:uh} proves Theorem~\ref{thm:strong}}.
\end{proof}

\section{A comparison to the quadrature approach}\label{subsec:rat-comparequad}

\cite{bolin2017numerical} proposed another method
which can be applied to
\kk{simulate} the solution~$u$ to \eqref{e:Lbeta}
numerically.
The approach therein is
to express the discretized equation \eqref{e:uh}
as $L_h^{\tilde{\beta}}L_h^{\lfloor \beta \rfloor} u_h = \white_h$,
where $\tilde{\beta} = \beta-\lfloor\beta\rfloor \in [0,1)$. 
Since $L_h^{\lfloor \beta \rfloor} u_h = f$ can be solved
by using non-fractional methods,
the focus was on the fractional case $\beta\in(0,1)$
when constructing the approximative solution.
From the Dunford--Taylor calculus \citep[{\S}IX.11]{yosida1995}
one has in this case the following representation
of the discrete inverse,
\begin{align*}
    L_h^{-\beta}
        &=
        \frac{\sin(\pi\beta)}{\pi}
        \int_0^{\infty} \lambda^{-\beta} 
        	\left(\lambda \, \kk{\mathrm{Id}_h} + L_h \right)^{-1} \, \rd \lambda.
\end{align*}
\cite{bonito2015} introduced a quadrature approximation $Q_{h,k}^\beta$
of this integral
after a change of variables $\lambda = e^{-2y}$
and based on an equidistant grid
%
for $y$ with step size~$k>0$, \kk{i.e.,
\begin{align*}
    Q_{h,k}^\beta := \frac{2 k \sin(\pi\beta)}{\pi} 
    \sum_{j=-K^{-}}^{K^{+}} e^{2\beta y_j} 
    	\left( \kk{\mathrm{Id}_h} + e^{2 y_j} L_h \right)^{-1}, 
    \quad\text{where}\quad  
    y_j := j k.
\end{align*}
Exponential convergence
of order $\cO\bigl( e^{-\pi^2/(2k)} \bigr)$}
of the operator $Q_{h,k}^\beta$
to the discrete fractional inverse $L_h^{-\beta}$
was proven for
$K^{-} := \left\lceil \frac{\pi^2}{4\beta k^2} \right\rceil$ 
and 
$K^{+} := \left\lceil \frac{\pi^2}{4(1-\beta) k^2} \right\rceil$. 

By calibrating the number of quadrature nodes with the
number of basis functions in the FEM, 
an explicit rate of convergence for the strong error 
of the approximation
$u_{h,k}^Q = Q_{h,k}^\beta \white_h$
%
%
was derived \citep[][Theorem 2.10]{bolin2017numerical}.
Motivated by the asymptotic convergence of the method,
\kk{it was suggested to choose
$k\leq - \tfrac{\pi^2}{4\beta\ln(h)}$
in order to balance the errors 
induced by the quadrature and by a FEM of mesh size $h$ 
\citep[][Table~1]{bolin2017numerical}. 
This corresponds} to a total number of
$K = K^- + K^+ + 1 > \tfrac{4\beta\ln(h)^2}{\pi^2(1-\beta)}$
quadrature nodes.
The analogous result for the degree $m$
of the approximation $u_{h,m}^R$ is given
in Remark~\ref{rem:calibrate-h-m},
suggesting the lower bound $m \geq \tfrac{\ln(h)^2}{\pi^2 (1-\beta)}$,
i.e., $K = 4\beta m$ asymptotically.

Furthermore, if we let \kk{$c_j := e^{2 y_j}$ and
\begin{align*}
    \lop^Q := \prod_{j=-K^{-}}^{K^{+}} c_j^{-\beta} \left( \mathrm{Id}_h + c_j L_h \right),
    \quad
    \rop^Q := \frac{2 k \sin(\pi\beta)}{\pi} 
    \sum_{i=-K^{-}}^{K^{+}} \prod_{j \neq i} c_j^{-\beta} 
    	\left( \mathrm{Id}_h + c_j L_h \right),
\end{align*}
we} find that the quadrature-based
approximation $u_{h,k}^Q$
can equivalently be defined
as the solution to the non-fractional SPDE
\begin{align}\label{e:quad-rat}
    \lop^Q u_{h,k}^Q  = \rop^{Q} \white_h
    \quad
    \text{in } \cD.
\end{align}

\begin{remark}\label{rem:quad}
A comparison of~\eqref{e:quad-rat} with~\eqref{e:uhr}
illustrates that $u_{h,k}^Q$
can be seen as a rational approximation
of degree $K^{-} + K^{+}$,
where the specific choice of the coefficients
is implied by the quadrature.
In combination with the remark above
that $K=4\beta m$ quadrature nodes are needed
to balance the errors,
this shows that the computational cost
for achieving a given accuracy
with the rational approximation from \S\ref{subsec:rat-approx}
is lower than with the quadrature method,
since $\beta > d/4$.
\end{remark}

\section{Parameter identifiability}\label{sec:measureproof}
This section contains the proof of Theorem~\ref{thm:measure}. For the proof, we will use the Feldman--H\'ajek theorem which we restate here from \citep[Theorem 2.25]{DaPrato2014} for convenience. 
%
\begin{theorem}[Felman--H\'ajek]\label{thm:stuart}
Two Gaussian measures $\mu_1 = \pN(m_1,\mathcal{C}_1)$ and $\mu_2 = \pN(m_2,\mathcal{C}_2)$ on a Hilbert space~$\mathcal{H}$ 
are either singular or equivalent. 
They are equivalent if and only if the following three conditions 
are satisfied:
\begin{enumerate}[label = \normalfont\Roman*.]
\item\label{stuart-1} $\operatorname{Im}\bigl(\mathcal{C}_1^{1/2}\bigr) 
	 = \operatorname{Im}\bigl(\mathcal{C}_2^{1/2}\bigr) := E$,
\item\label{stuart-2} $m_1-m_2 \in E$,
\item\label{stuart-3} the operator $T:= 
\bigl(\mathcal{C}_1^{-1/2}\mathcal{C}_2^{1/2}\bigr)
\bigl(\mathcal{C}_1^{-1/2}\mathcal{C}_2^{1/2}\bigr)^* - I$ 
is Hilbert-Schmidt in $\bar{E}$, where $*$ denotes the $\cH$-adjoint operator, 
and $I$ the identity on $\cH$.  
\end{enumerate}
\end{theorem}

\begin{proof}[Proof of Theorem~\ref{thm:measure}]
Since the two Gaussian measures have the same mean, 
we only have to verify conditions~\ref{stuart-1} and~\ref{stuart-3} 
of Theorem~\ref{thm:stuart}.

We first prove that condition~\ref{stuart-1} 
can hold only if $\beta_1=\beta_2$. 
To this end, we use the equivalence of 
condition~\ref{stuart-1} with the existence 
of two constants $c', c''>0$ such that 
\begin{align}\label{eq:stuart-lemma}  
	\scalar{v, \cC_1 v}{L_2(\cD)} 
	\leq c' 
	\scalar{v, \cC_2 v}{L_2(\cD)}  
	\quad 
	\text{and} 
	\quad 
	\scalar{v, \cC_2 v}{L_2(\cD)} 
	\leq c'' 
	\scalar{v, \cC_1 v}{L_2(\cD)},   
\end{align} 
see, e.g., \citep[][Lemma 6.15]{stuart2010}, 
where in our case 
$$
\cC_i := \cQ_i^{-1} = \tau_i^{-2}(\kappa_i^2 - \Delta)^{-2\beta_i}, \qquad i\in\{1,2\}.
$$ 

Let $\lambda_j^\Delta$, $j\in\bbN$, 
denote the positive eigenvalues 
(in nondecreasing order, counting multiplicity)  
of the Dirichlet or Neumann 
Laplacian $-\Delta\from\scrD(\Delta)\to L_2(\cD)$,  
where the type of homogeneous boundary conditions 
is the same 
as for $L_1$ and $L_2$. 
By Weyl's law~\eqref{e:lambdaj}, there exist 
constants $\underline{c},\bar{C}>0$ such that 
\[
	\underline{c} \, j^{2/d} \leq \lambda_j^{\Delta} \leq \bar{C} j^{2/d} 
	\qquad 
	\forall j \in \bbN. 
\] 
Furthermore, we let $\{e_j\}_{j\in\bbN}$ 
denote a system of eigenfunctions 
corresponding to $\bigl\{\lambda_j^\Delta \bigr\}_{j\in\bbN}$ 
which is orthonormal in $L_2(\cD)$. 

Now assume that $\beta_2 > \beta_1$ 
and let $j_0\in\bbN$ be sufficiently large such that 
$\kappa_1^2 < \bar{C} j_0^{2/d}$. 
Then, we have 
\[
	\frac{(\kappa_2^2 + \lambda_j^{\Delta})^{2\beta_2}}{
		(\kappa_1^2 + \lambda_j^{\Delta})^{2\beta_1}} 
	> 
	\frac{\underline{c}^{2\beta_2}}{(2\bar{C})^{2\beta_1}} \, 
	j^{4(\beta_2-\beta_1)/d} 
	\qquad 
	\forall j\in\bbN, \
	j \geq j_0. 
\]
For any $N\in\bbN$, 
we can thus  
choose $j_* = j_*(N)\in\bbN$ sufficiently large 
such that 
\[
	\scalar{e_{j_*}, \cC_1 e_{j_*}}{L_2(\cD)}  
	= 
	\tau_1^{-2} (\kappa_1^2 + \lambda_{j_*}^{\Delta})^{-2\beta_1} 
	> 
	N 
	\tau_2^{-2} (\kappa_2^2 + \lambda_{j_*}^{\Delta})^{-2\beta_2} 
	= 
	N 
	\scalar{e_{j_*}, \cC_2 e_{j_*}}{L_2(\cD)}, 
\]
in contradiction with the 
first relation in~\eqref{eq:stuart-lemma}, and  
$\mu_1, \mu_2$ are not equivalent 
if $\beta_1\neq \beta_2$. 
Furthermore, condition~\ref{stuart-1} 
is satisfied if $\beta_1=\beta_2 = \beta>d/4$, since then, 
for all $v\in L_2(\cD)$, 
\begin{align*} 
	\scalar{v, \cC_1 v}{L_2(\cD)}  
	&= 
	\sum_{j\in\bbN} \tau_1^{-2} 
	(\kappa_1^2 +\lambda_j^\Delta)^{-2\beta} \scalar{v,e_j}{L_2(\cD)}^2 \\
	&\leq 
	\tau_1^{-2} \tau_2^2 
	\left(\min\left\{ 1, \kappa_1^{2}\kappa_2^{-2} \right\} \right)^{-2\beta}
	\sum_{j\in\bbN} \tau_2^{-2} 
	(\kappa_2^2 +\lambda_j^\Delta)^{-2\beta} \scalar{v,e_j}{L_2(\cD)}^2 \\
	&= 
	\tau_1^{-2} \tau_2^2  
	\max\left\{ 1, \kappa_1^{-4\beta}\kappa_2^{4\beta} \right\}
	\scalar{v, \cC_2 v}{L_2(\cD)},   
\end{align*} 
and, similarly, 
$\scalar{v, \cC_2 v}{L_2(\cD)}  
	\leq
	\tau_2^{-2} \tau_1^2  
	\max\bigl\{ 1, \kappa_2^{-4\beta}\kappa_1^{4\beta} \bigr\} \scalar{v, \cC_1 v}{L_2(\cD)}$. 
Thus,~\eqref{eq:stuart-lemma} and condition~\ref{stuart-1} 
of Theorem~\ref{thm:stuart} hold. 

Assuming that $\beta_1=\beta_2 = \beta>d/4$, 
it remains now to show that condition~\ref{stuart-3} 
of Theorem~\ref{thm:stuart}
is satisfied if and only if $\tau_1=\tau_2$.  
To this end, we first note that the operator 
$T := \mathcal{C}_1^{-1/2}\mathcal{C}_2 
\mathcal{C}_1^{-1/2} - I$ has eigenfunctions 
$\{e_j\}_{j\in\bbN}$ and eigenvalues 
\[
	\tau_1^2 \tau_2^{-2} 
	(\kappa_1^2 + \lambda_j^\Delta)^{2\beta}
	(\kappa_2^2 + \lambda_j^\Delta)^{-2\beta} 
	-1 , 
	\qquad 
	j\in\bbN. 
\]
Therefore, $T$ is Hilbert--Schmidt in $\bar{E}$ 
if and only if 
\begin{equation}\label{eq:app:HS-T} 
	\sum_{j\in\bbN} 
	\left( 
	\tau_1^2 \tau_2^{-2} 
	(\kappa_1^2 + \lambda_j^\Delta)^{2\beta}
	(\kappa_2^2 + \lambda_j^\Delta)^{-2\beta} 
	-1
	\right)^2 < \infty. 
\end{equation} 
Since $x\mapsto (1+x)^{1/(2\beta)}$ 
is monotonically increasing in $x>0$, 
again by the Weyl asymptotic, 
for any $\varepsilon_0 > 0$, we can find 
an index 
$j_0\in\bbN$ such that 
\begin{equation}\label{eq:app:j0} 
	\frac{\kappa_2^2}{\lambda_j^\Delta} + 1 
	\leq 
	(1+\varepsilon_0)^{1/(2\beta)} 
	\qquad 
	\forall j\in\bbN, \ 
	j \geq j_0. 
\end{equation} 
Assume that $\tau_1\neq\tau_2$ and without loss of generality 
let $\tau_1> \tau_2$. 
Then pick $\varepsilon_0>0$ such that 
$\tau_1^{2} \tau_2^{-2} \geq 1 + 2\varepsilon_0$, 
and $j_0\in\bbN$ such that~\eqref{eq:app:j0} holds. 
These choices give 
\begin{align*} 
	\tau_1^2 \tau_2^{-2} 
	\left( 
	\frac{\kappa_1^2 + \lambda_j^\Delta}{\kappa_2^2 + \lambda_j^\Delta} 
	\right)^{2\beta} 
	\geq 
	\tau_1^2 \tau_2^{-2} 
	(\kappa_2^2/ \lambda_j^\Delta +1 )^{-2\beta} 
	\geq 
	(1 + 2\varepsilon_0) 
	(1 + \varepsilon_0)^{-1}>1,
\end{align*} 
for all $j\in\bbN$ with $j \geq j_0$. 
Thus, the series in~\eqref{eq:app:HS-T} 
is unbounded, 
\begin{align*} 
	\sum_{j\in\bbN} 
	\left( 
	\tau_1^2 \tau_2^{-2} 
	(\kappa_1^2 + \lambda_j^\Delta)^{2\beta}
	(\kappa_2^2 + \lambda_j^\Delta)^{-2\beta} 
	-1
	\right)^2 
	&\geq 
	\sum_{j\geq j_0} 
	\left( 
	(1 + 2\varepsilon_0) 
	(1 + \varepsilon_0)^{-1} 
	-1
	\right)^2 \\
	&= 
	\sum_{j\geq j_0} 
	\varepsilon_0^2 
	\left( 1 + \varepsilon_0  
	\right)^{-2}  
	=\infty .
\end{align*} 
We conclude that condition~\ref{stuart-3} 
of Theorem~\ref{thm:stuart} 
is not satisfied if $\tau_1\neq\tau_2$. 

Finally, let $\beta_1=\beta_2=\beta$, $\tau_1=\tau_2$ 
and assume without loss of generality that 
$\kappa_2 > \kappa_1$ (if $\kappa_1=\kappa_2$, 
\eqref{eq:app:HS-T} is evident). 
By the mean value theorem, applied for the 
function $x\mapsto x^{2\beta}$, for every $j\in\bbN$, 
there exists 
$\widetilde{\kappa}_j \in (\kappa_1, \kappa_2)$ 
such that 
\[
	(\kappa_2^2 + \lambda_j^\Delta)^{2\beta} - (\kappa_1^2 + \lambda_j^\Delta)^{2\beta} 
	= 
	2\beta (\widetilde{\kappa}_j^2 + \lambda_j^\Delta)^{2\beta-1} (\kappa_2^2-\kappa_1^2). 
\]
Hence, we can bound the series in~\eqref{eq:app:HS-T} 
as follows, 
\begin{align*} 
	\sum_{j\in\bbN} 
	&\left( 
	\frac{ (\kappa_1^2 + \lambda_j^\Delta)^{2\beta} - 
		(\kappa_2^2 + \lambda_j^\Delta)^{2\beta} } { 
	(\kappa_2^2 + \lambda_j^\Delta)^{2\beta} }
	\right)^2
	= 
	4 \beta^2 (\kappa_2^2-\kappa_1^2)^2 
	\sum_{j\in\bbN} 
	\left( 
	\frac{
		(\widetilde{\kappa}_j^2 + \lambda_j^\Delta)^{2\beta-1} }{
	(\kappa_2^2 + \lambda_j^\Delta)^{2\beta}} 
	\right)^2 \\
	&\leq 
	4 \beta^2 (\kappa_2^2-\kappa_1^2)^2  
	\sum_{j\in\bbN} 
	(\widetilde{\kappa}_j^2 + \lambda_j^\Delta)^{-2}
	\leq 
	4 \beta^2 (\kappa_2^2-\kappa_1^2)^2  
	\underline{c}^{-2} 
	\sum_{j\in\bbN} 
	j^{-4/d} < \infty.   
\end{align*} 
Here, $\sum_{j\in\bbN} j^{-4/d}$ converges, 
since $4/d > 1$ for $d\in\{1,2,3\}$. 
This proves equivalence of 
the Gaussian measures if 
$\beta_1=\beta_2$ and $\tau_1=\tau_2$. 
\end{proof}
\end{appendix}

\bibliographystyle{apa}
\bibliography{fractional-bib}

\end{document}